\documentclass[11pt]{amsart}

\usepackage[a4paper,centering]{geometry}
\usepackage[latin1]{inputenc}
\usepackage{aeguill}
\usepackage{enumerate}
\usepackage{epsfig}
\usepackage{graphicx}
\usepackage{psfrag}
\usepackage{color}
\usepackage{url}
\usepackage{algorithm}
\usepackage{algorithmic}

\newtheorem{theorem}{Theorem}

\newtheorem{lemma}[theorem]{Lemma}

\def\cal{\mathcal}

\def\pp{\mathcal{P}}
\def\lpp{\mathcal{L}}
\def\smallsetminus{\setminus}

\newcommand{\CA}[1]{{\mathcal C}_{#1}}
\newcommand{\MS}[1]{{\mathcal M}_{#1}}
\newcommand{\DS}[1]{{\mathcal D}_{#1}}
\newcommand{\AT}{Axiomatization Theorem}
\newcommand{\GRT}{Geometric Representation Theorem}
\newcommand{\PL}{Pumping Lemma}

\def\lgx{\ell(\gamma,x)} 

\def\MOB{\cal M}
\def\OPC{\overline{\cal M}}
\def\tba{{\tiny TBA}}
\def\tba{{\tiny NC}}
\def\nco{{\tiny NC}}

\def\gp{\gamma'}
\def\gpp{\gamma''}
\def\gpstar{\gp_*}
\def\gppstar{\gpp_*}
\def\gcurve{$\gamma$-curve} 
 
\def\tauG{\Gamma_\tau} 
\def\tauGp{\Gamma'_\tau} 
\def\tauGpp{\Gamma''_\tau} 

\newcommand{\NSOMPA}[1]{d^S_{#1}}
\newcommand{\NOMPA}[1]{d_{#1}}
\newcommand{\NSPPA}[1]{t^S_{#1}}
\newcommand{\NSPA}[1]{p^S_{#1}}
\newcommand{\NSMA}[1]{q^S_{#1}}
\newcommand{\NSAMA}[1]{r^S_{#1}}
\newcommand{\NSPCh}[1]{\rho^S_{#1}}
\newcommand{\NSMCh}[1]{\tau^S_{#1}}
\newcommand{\NPPA}[1]{t_{#1}}
\newcommand{\NPA}[1]{p_{#1}}
\newcommand{\NMA}[1]{q_{#1}}
\newcommand{\NAMA}[1]{r_{#1}}
\newcommand{\NPCh}[1]{\rho_{#1}}
\newcommand{\NMCh}[1]{\tau_{#1}}

\newcommand{\GSP}[2]{g^S(#1,#2)}
\newcommand{\GP}[2]{g(#1,#2)}
\newcommand{\GSM}[2]{h^S(#1,#2)}
\newcommand{\GM}[2]{h(#1,#2)}
\renewcommand{\GSP}[2]{g^S_{#2}(#1)}
\renewcommand{\GP}[2]{g_{#2}(#1)}
\renewcommand{\GSM}[2]{h^S_{#2}(#1)}
\renewcommand{\GM}[2]{h_#2(#1)}

\title[\today]{On the number of simple arrangements of five double pseudolines} 
\author{Julien Fert\'e}
\author{Vincent Pilaud}
\author{Michel Pocchiola}
\address{Julien Fert{\'e}\\
Laboratoire d'Informatique Fondamentale\\
Universit{\'e} de Provence\\
Marseille\\
France}
\email{julien.ferte@lif.univ-mrs.fr}
\address{Vincent Pilaud\\
Universit{\'e} Pierre et Marie Curie\\
{\'E}quipe Combinatoire et Optimisation\\
Paris\\
France}
\email{vpilaud@math.jussieu.fr}
\address{Michel Pocchiola\\
Universit{\'e} Pierre et Marie Curie\\
{\'E}quipe Combinatoire et Optimisation\\
Paris\\
France}
\email{pocchiola@math.jussieu.fr}

\graphicspath{{figures/}}

\begin{document}
\maketitle

\begin{abstract} We describe an incremental algorithm to enumerate the isomorphism classes of double pseudoline arrangements.
The correction of our algorithm is based on the connectedness under mutations of the spaces of one-extensions of double pseudoline arrangements, 
proved in this paper. Counting results derived from an implementation of our algorithm are also reported.
\end{abstract}


\section{Introduction}
An \emph{arrangement of double pseudolines} 
is a finite set of separating simple closed curves embedded in a real two-dimensional projective plane 
such that any two  curves have exactly four intersection points, cross transversally at these points, and 
induce a cell decomposition of their underlying projective plane.  
Two arrangements of double pseudolines are said to be  \emph{isomorphic} 
if one is the image of the other by a homeomorphism of their underlying projective planes.
There is a unique isomorphism class of arrangements of two double pseudolines and 
Figure~\ref{fulllist} depicts representatives of the thirteen isomorphism 
classes of simple arrangements of three double pseudolines where, as usual, a simple arrangement is an arrangement 
where no three curves meet at the same point.

\begin{figure}[!htb]
\footnotesize
\def\factor{0.15315015000023}
\centering
\psfrag{8}{} \psfrag{7}{} \psfrag{6}{} \psfrag{5}{} \psfrag{4}{} \psfrag{3}{} \psfrag{2}{}
\psfrag{A}{$04$}
\psfrag{B}{$07$}
\psfrag{C}{$18$}
\psfrag{CCone}{$18_1$}
\psfrag{D}{$25$} 
\psfrag{F}{$07$} 
\psfrag{G}{$37$}
\psfrag{H}{$15$}
\psfrag{HCone}{$15_1$}
\psfrag{J}{$43$}
\psfrag{JCone}{$43_1$}
\psfrag{K}{$25$}
\psfrag{KCone}{$25_1$}
\psfrag{L}{$33$}
\psfrag{LCone}{$33_1$}
\psfrag{M}{$32$}
\psfrag{N}{$25$}
\psfrag{Nstar}{$25^*$}
\psfrag{NstarCone}{$25^*_1$}
\psfrag{NstarCtwo}{$25^*_2$}
\psfrag{O}{$32$}
\psfrag{OCone}{$32_1$}
\psfrag{OCtwo}{$32_2$}
\psfrag{P}{$22$}
\psfrag{PCone}{$22_1$}
\psfrag{Q}{$25$}
\psfrag{R}{$36$}
\psfrag{Z}{$64$}
\psfrag{o24}{24} \psfrag{o12}{12} \psfrag{o2}{2} \psfrag{o4}{4} \psfrag{o6}{6} \psfrag{o1}{1}
\includegraphics[width = \factor\linewidth]{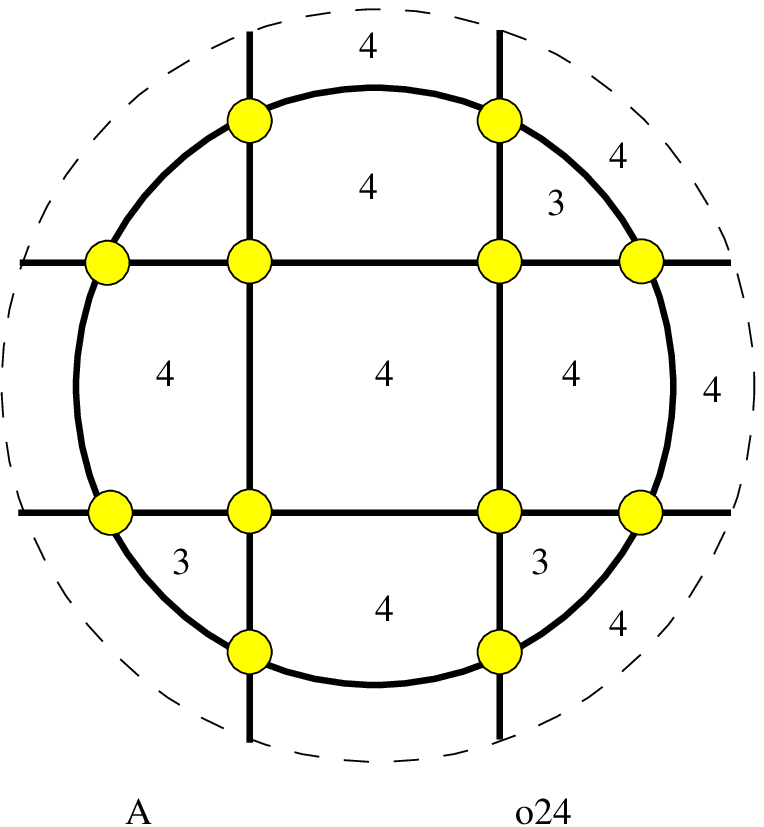}
\includegraphics[width = \factor\linewidth]{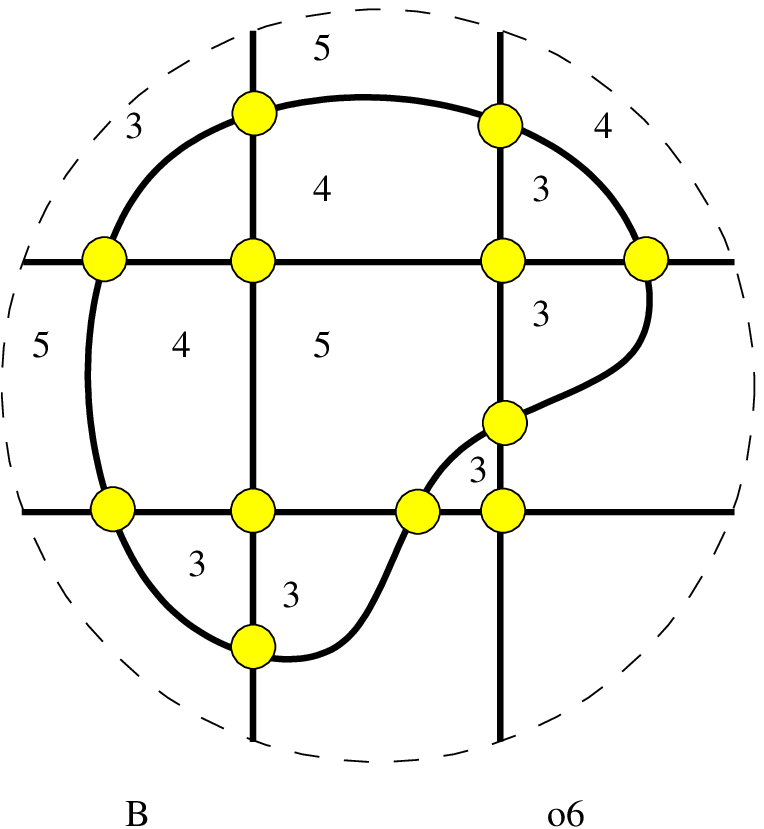}
\includegraphics[width = \factor\linewidth]{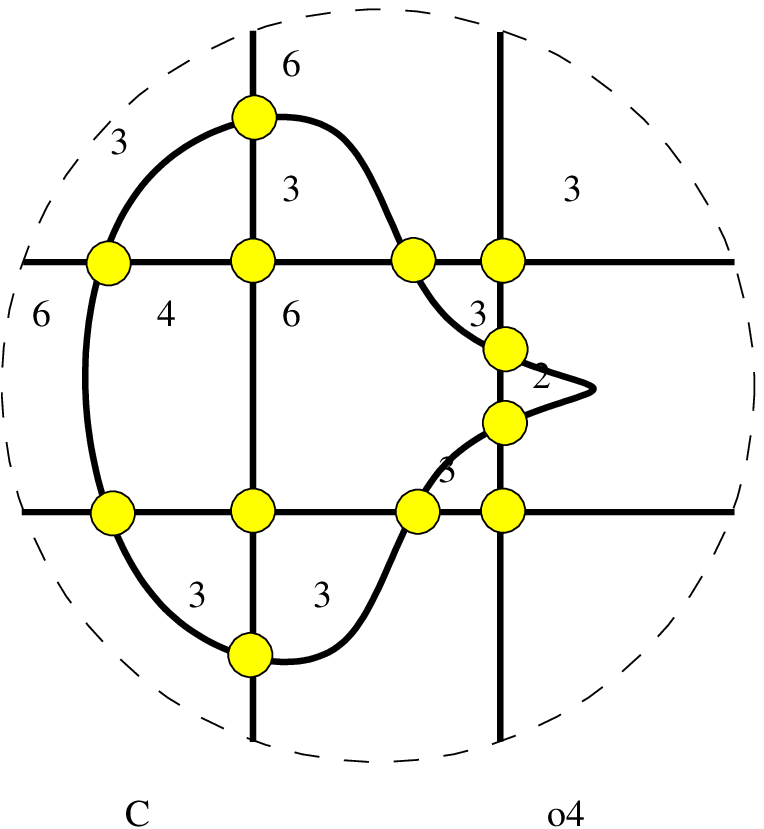}
\includegraphics[width = \factor\linewidth]{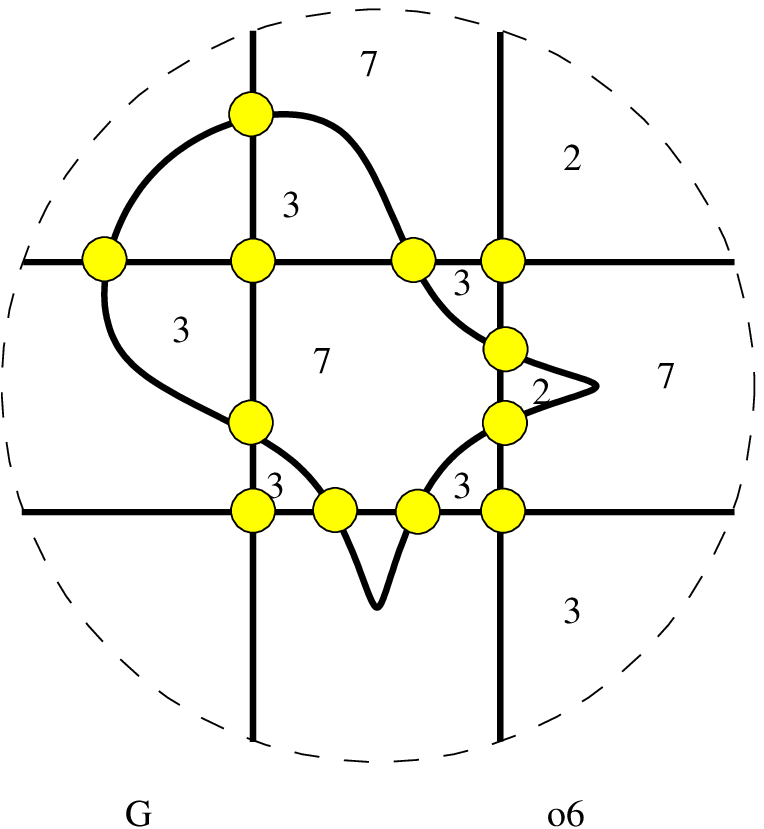}
\includegraphics[width = \factor\linewidth]{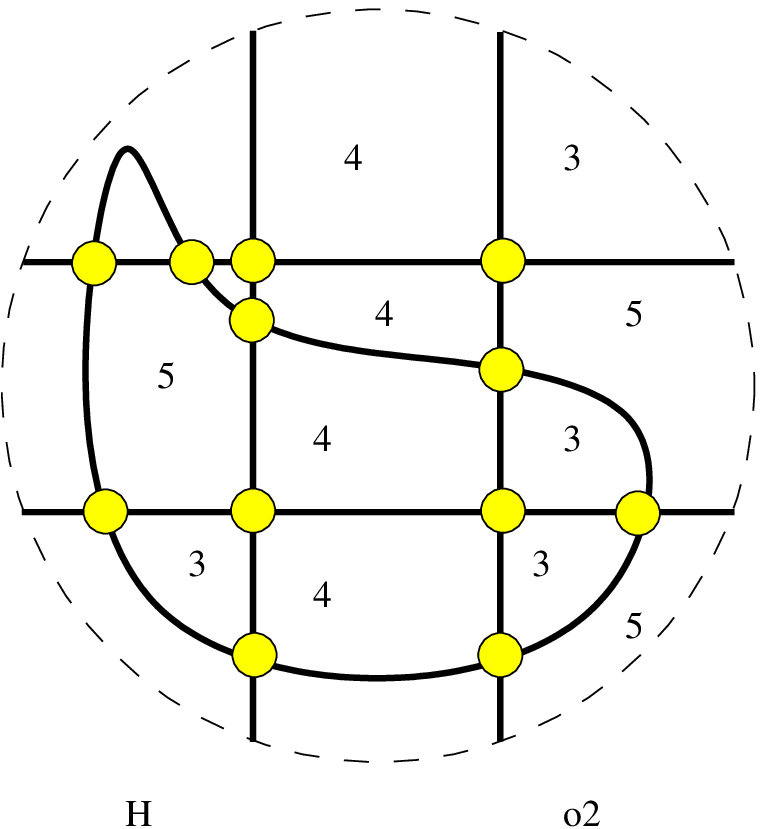}
\includegraphics[width = \factor\linewidth]{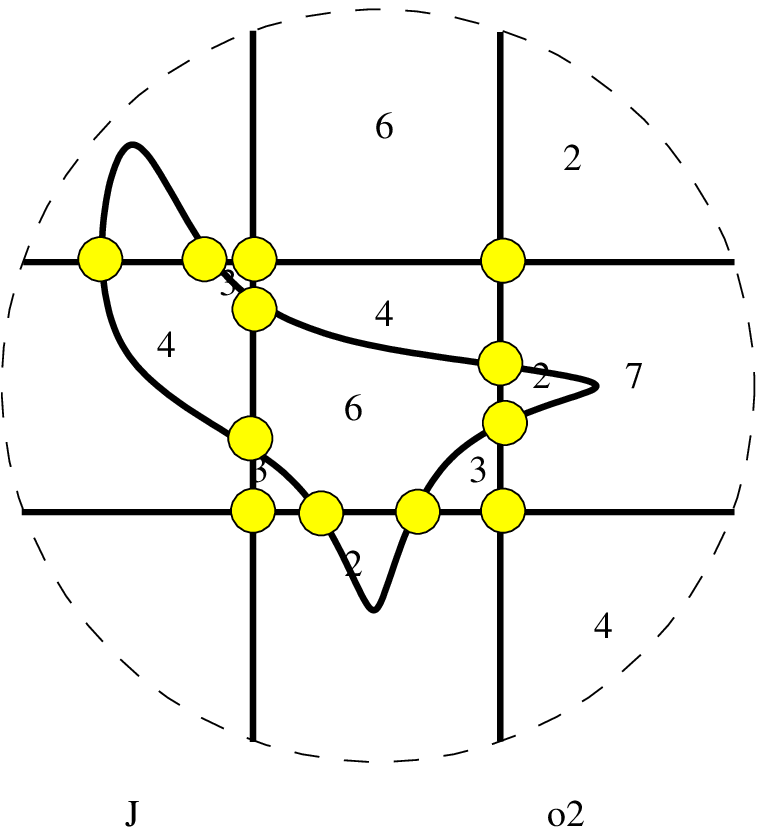}
\includegraphics[width = \factor\linewidth]{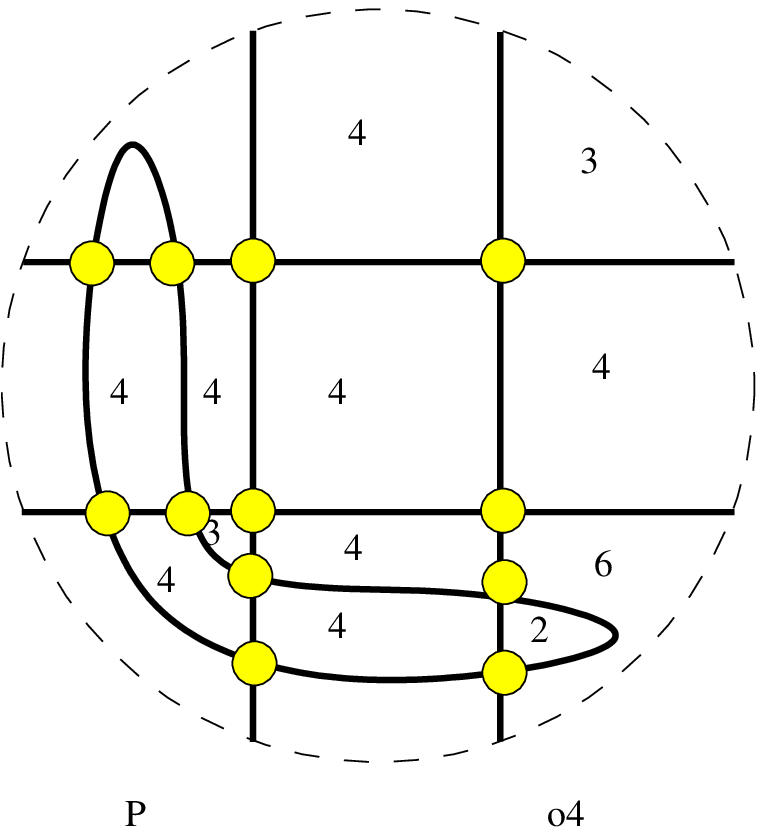}
\includegraphics[width = \factor\linewidth]{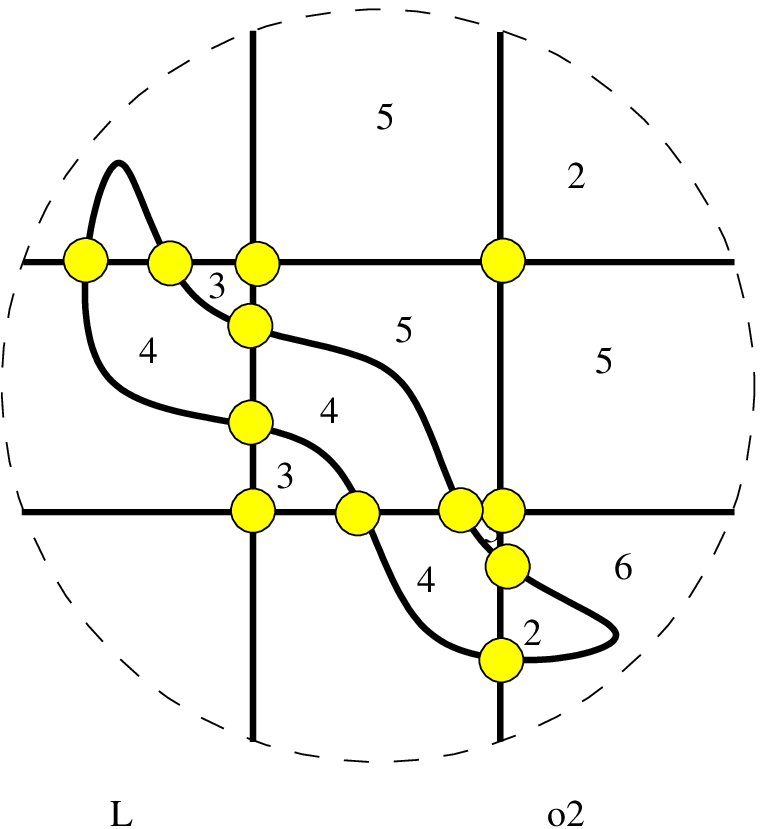}
\includegraphics[width = \factor\linewidth]{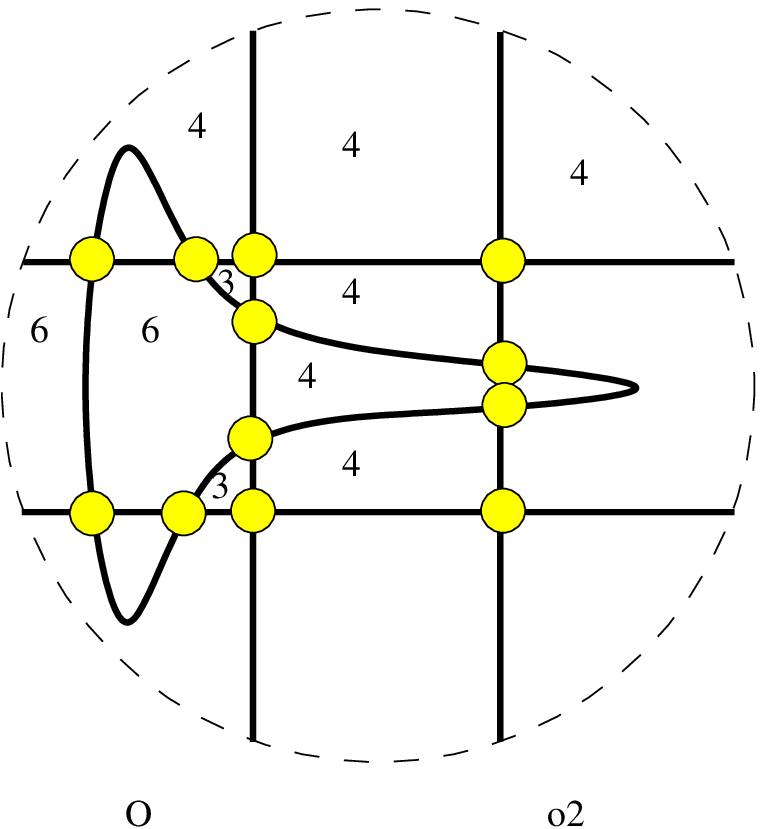}
\includegraphics[width = \factor\linewidth]{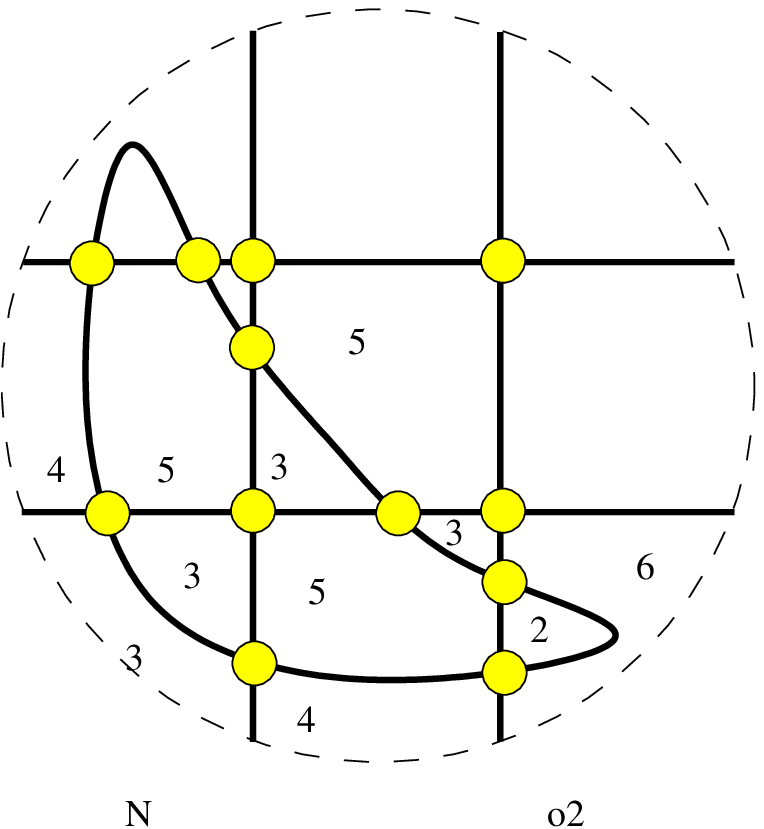}
\includegraphics[width = \factor\linewidth]{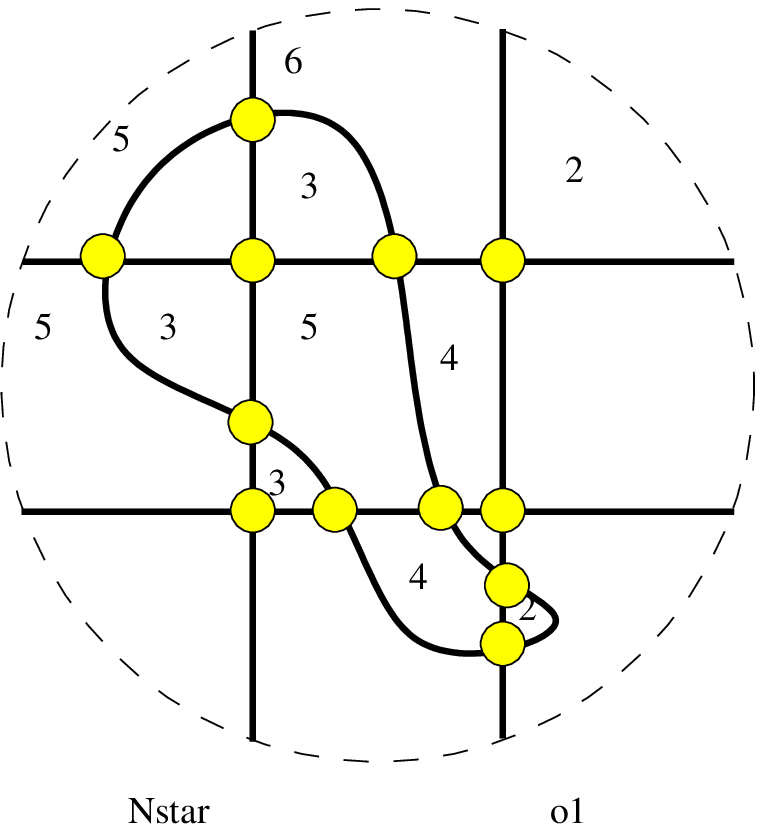}
\includegraphics[width = \factor\linewidth]{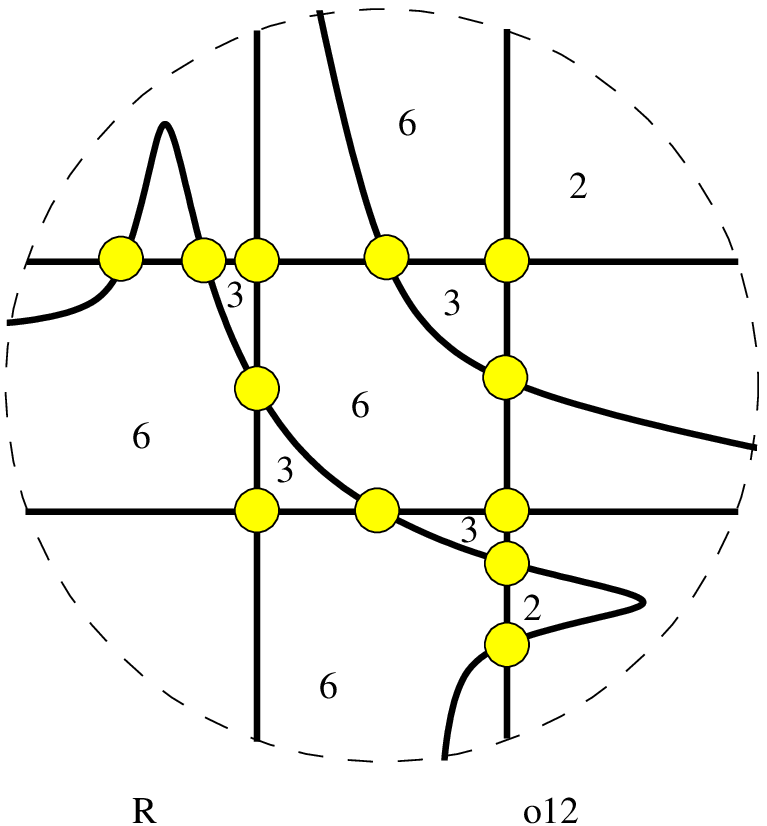}
\includegraphics[width = \factor\linewidth]{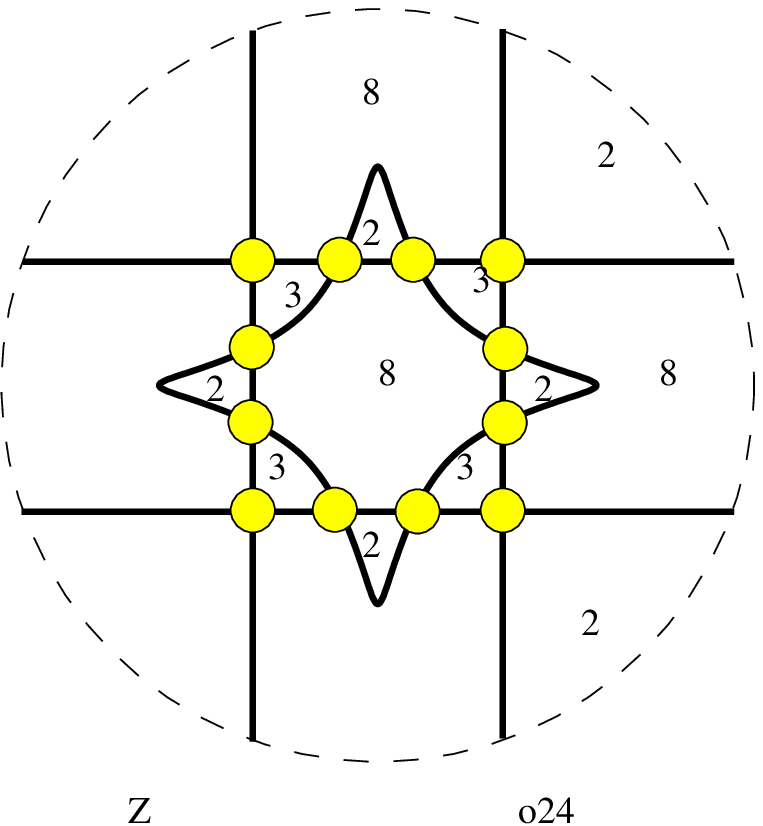}
\caption{Representatives of the thirteen isomorphism classes of simple arrangements of three double pseudolines. 
Each arrangement is labelled at its left bottom corner with the $2$-sequence of its numbers of two-cells of size~2 and~3,
and at its right bottom corner with the order of its automorphism group. 
\label{fulllist}}
\end{figure}

In this paper we prove that the one-extension spaces of double pseudoline arrangements are connected under mutations.
This connectedness result yields an incremental enumeration algorithm that takes as input the class of isomorphism classes of 
arrangements of $n$  double pseudolines and returns as output the class of isomorphism classes of 
arrangements of $n+1$ double pseudolines.
The key feature of this enumeration algorithm is that its working space is proportional to the size of the input and {\it not} to the size of the output;
it turns out to be sufficient in practice for the enumeration of simple arrangements of five 
double pseudolines using a relatively modest amount of process time: we establish in this paper that their number is $181\,403\,533$.

Double pseudoline arrangements have been introduced recently by Luc Habert and the third author of the paper
as a combinatorial abstraction of finite families of disjoint convex bodies of real two-dimensional projective geometries, 
in connection with their study of primitive operations for efficient visibility graph algorithms for planar convex shapes.
More precisely they have shown that the class of isomorphism classes of double pseudoline arrangements coincides  with the class 
of isomorphism classes of dual arrangements of finite families of pairwise disjoint convex bodies of projective geometries, 
 and they have exploited this latter result to give an axiomatic characterization of the class of {\it chirotopes} of finite families of 
pairwise disjoint oriented convex bodies of projective geometries very similar to classical axiomatic characterizations of the
class of chirotopes of finite families of points of projective geometries known under the generic name of (rank three) oriented matroids.
 (The {\it chirotope} of a finite family of pairwise disjoint oriented convex bodies of a projective geometry is defined as 
the map that assigns to each triple of convex bodies the set of relative positions of the bodies of the triple with respect to the lines of the projective geometry---the primitive operations to which we referred above consist precisely in the evaluation of the chirotope of the input configuration to the visibility graph algorithm.) 
These results provide our main motivation to enumerate double pseudoline arrangements
(cf. ~\cite{G-hp-adp-06}).

The paper is organized as follows. 
In Section~\ref{sec:preliminaries} we recall definitions and structural properties of arrangements of double 
pseudolines that on one hand motivate the paper and on the other hand are used subsequently in the technical developments. 
These structural properties are {\it the axiomatic characterization of the class 
of isomorphim classes of double pseudoline arrangements in terms of chirotopes}, {\it the connectedness of mutation graphs}, and 
the so-called {\it \PL} and {\it \GRT}. 
Still in Section~\ref{sec:preliminaries} we establish an enhanced version of the {\it \PL} using the {\it \GRT}.  
In Section~\ref{sec:connectedness}  we prove the connectedness result mentioned above using the enhanced version of the {\it \PL}. 
In Section~\ref{sec:incremental} we describe our incremental algorithm; in particular we explain how to add a double pseudoline to an arrangement of double pseudolines,
 a procedure that we believe to be of independant interest.  In Section~\ref{sec:results} we report some of the counting results for simple as well as non-simple arrangements of 
double pseudolines 
that we have obtained from our implementation of the incremental algorithm. 
(Counting results for the subclass of the so-called M{\"o}bius arrangements are also reported; this subclass captures exactly the class of chirotopes of affine configurations 
of disjoint convex bodies.) 
Finally in Section~\ref{sec:further} we conclude by  a short series of questions and possible developments suggested by this research. 
Throughout the paper we assume the reader to be familiar with the basic terminology of pseudoline arrangements~\cite{g-as-72,blswz-om-99,b-com-06,g-cpl-09,g-pa-04}.


\section{Preliminaries}\label{sec:preliminaries}
\subsection{Projective arrangements of double pseudolines}
Let $\mathcal{P}$ be a real two-dimen\-sio\-nal projective plane. In our drawings we will represent it by  a circular diagram  with antipodal boundary points identified.
A simple closed curve in $\mathcal{P}$ is a \emph{pseudoline} if it is non-separating (or equivalently, non-contractible), 
and a \emph{double pseudoline} otherwise.  
The complement of a pseudoline~$\gamma$ has one connected component: a topological disk~$\DS{\gamma}$. 
The complement of a double pseudoline~$\gamma$ has two connected components: a M{\"o}bius strip $\MS{\gamma}$ 
and a topological disk~$\DS{\gamma}.$ 
Observe that a pseudoline avoiding a double pseudoline is necessarily included in the M{\"o}bius strip surrounded by the double pseudoline
 (Fig.~\ref{projectiveplane}). 

\begin{figure}[!htb]
\centering
\includegraphics[width=0.90\linewidth]{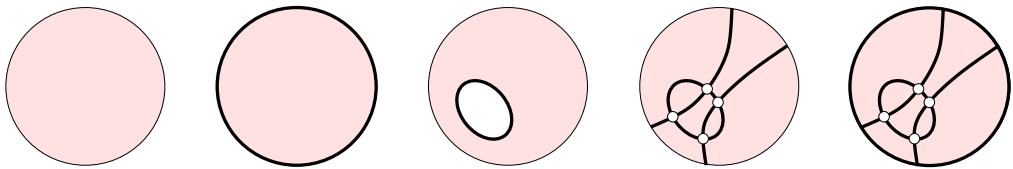}
\caption{A projective plane represented by a circular diagram with antipodal
points identified, a pseudoline, a double pseudoline, an arrangement of two double pseudolines, and 
a one-marked arrangement of two double pseudolines.
\label{projectiveplane}}
\end{figure}

An \emph{arrangement of double pseudolines} in $\mathcal{P}$ is a finite set of double pseudolines such that 
any two double pseudolines have exactly four intersection points, cross transversally at these points, and induce a cell decomposition of $\mathcal{P}$ 
(Fig.~\ref{projectiveplane}). 
We call it  {\it full} if the intersection of the topological disks surrounded by its double pseudolines is empty. 
The {\it order} of an arrangement is its number of  curves. 
As usual a {\it simple} arrangement is an arrangement where no three curves meet at the same point.
A {\it marked} arrangement of double pseudolines is an arrangement of double pseudolines augmented with an arrangement of pseudolines 
such that any pseudoline avoids one double pseudoline of the arrangement and intersects the other ones in exactly two points.
A {\it T-marked} arrangement of double pseudolines is an arrangement of double pseudolines augmented with an arrangement of pseudolines 
such that any pseudoline intersects one double pseudoline of the arrangement in exactly one point and intersects the other ones in exactly two points.

Two arrangements are \emph{isomorphic} if there is a homeomorphism of their underlying projective planes that sends one arrangement onto the other. 
The reader will easily check that there is a unique isomorphism class of arrangements of two double pseudolines, a unique isomorphism class of one-marked 
arrangements of two double pseudolines,  and 
(less easily) that there are thirteen isomorphism classes of
simple arrangements of three double pseudolines of which only one is full. They are depicted in Figure~\ref{fulllist}.
In this paper, we are interested in enumerating isomorphism classes of arrangements of double pseudolines. 
The number of isomorphism classes of arrangements of order $n$ will be denoted by $\NPA{n}$, the 
number of isomorphism classes of full arrangements of order $n$  will be denoted by~$\NPPA{n}$, and the number of isomorphism classes of
one-marked arrangements of order $n$ will be denoted by $\NOMPA{n}.$  
The corresponding number for simple arrangements will be denoted by $\NSPA{n}$, $\NSPPA{n}$ and $\NSOMPA{n}.$. Thus 
$\NPA{2} = \NSPA{2} = \NSOMPA{2}= 1$, $\NPPA{2} = \NSPPA{2} = 0$, $\NSPA{3} = 13$ and $\NSPPA{3} =1$.

\subsection{Chirotopes}
An {\it indexed oriented arrangement} is an arrangement whose curves are oriented and one-to-one indexed with some indexing set.
The \emph{chirotope} of an indexed oriented arrangement is the application that assigns to each triple of indices 
the isomorphism class of the subarrangement indexed by this triple.

As for pseudoline arrangements, the isomorphism class of an indexed oriented arrangement of double pseudolines only depends on its chirotope. 
Furthermore, given an application $\chi$ that assigns to each triple of indices the isomorphism class of an oriented arrangement of double pseudolines 
indexed by this triple, the two following properties are equivalent
\begin{enumerate}
\item $\chi$ is the chirotope of an indexed oriented arrangement,
\item the restriction of $\chi$ to the set of triples of any subset of at most five indices is the chirotope of an indexed oriented arrangement~\cite{G-hp-adp-06}.
\end{enumerate}

This result, called the {\it \AT} for double pseudoline arrangements, provides a strong motivation for enumerating arrangements of at most five double pseudolines. In particular 
the number $\NPCh{n}$ of chirotopes on $n$ double pseudolines (on a given indexing set) is given by the sum   
\begin{equation}
\sum_{k\neq 0} \frac{n! 2^n }{k} \GP{k}{n}
\end{equation}
where $\GP{k}{n}$ is the number of arrangements of $n$ double pseudolines with automorphism groups of order $k$; thus $\sum_{k\neq 0} \GP{k}{n} = \NPA{n}$.  
The corresponding numbers for simple arrangements will be denoted by $\NSPCh{n}$ and $\GSP{k}{n}.$ 
For example we read from the data of Figure~\ref{fulllist} that  
$\GSP{1}{3}=1$, $\GSP{2}{3}=5$, $\GSP{4}{3}= 2$, $\GSP{6}{3}=2$, $\GSP{12}{3} =1$, 
$\GSP{24}{3} =1$ and that $\GSP{k}{3} = 0$ for $k\notin \{1,2,4,6,12,24\}$; 
consequently there are $\NSPCh{3} = 214$ simple chirotopes on three double pseudolines.

\subsection{M{\"o}bius arrangements}
Let $\MOB$ be a M{\"o}bius strip and let $\OPC = \MOB \cup \{\infty \}$ be its one-point compactification.  
An arrangement of double pseudolines in $\MOB$ is an arrangement of double pseudolines in $\OPC$ with the property that the intersection of
the topological disks surrounded by the double pseudolines of the arrangement is nonempty and 
contains the point at infinity $\infty$. 
An indexed oriented M{\"o}bius arrangement is a M{\"o}bius arrangement whose  
double pseudolines are one-to-one indexed  and  oriented;  the arrangement is called acyclic if the orientations of the double pseudolines are coherent, in the sense that the double pseudolines are homotopic as oriented curves{}\footnote{In other words the curves are oriented according to the choice of a generator of the (infinite cyclic) 
fundamental group of the underlying M{\"o}bius strip.}.

As for projective double pseudoline arrangements (i) two M{\"o}bius arrangements are called isomorphic 
if one is the image of the other by a homeomorphism of their underlying M{\"o}bius strips, and (ii) 
the chirotope of an indexed oriented M{\"o}bius arrangement
is defined as the map that assigns to each triple of indices the isomorphism class of the (indexed and oriented) subarrangement indexed by this triple.  
The {\it \AT} for projective arrangements extends word for word for M{\"o}bius arrangements as well as acyclic M{\"o}bius arrangements.

We denote by $\NMA{n}$ the number of isomorphism classes of M{\"o}bius arrangements of order $n$,
by $\NAMA{n}$ the number of isomorphism classes of acyclically oriented M{\"o}bius arrangements of order $n$,
 and by $\NMCh{n}$ the number of isomorphism classes of indexed oriented 
acyclic M{\"o}bius arrangements of order $n$. 
The number $\NMCh{n}$ of acyclic M{\"o}bius chirotopes on $n$ double pseudolines is given by the sum   
\begin{equation}
\sum_{k\neq 0} \frac{n! 2 }{k} \GM{k}{n}
\end{equation}
where $\GM{k}{n}$ is the number of M{\"o}bius arrangements of $n$ double pseudolines with automorphism groups of order $k$.  
The corresponding numbers for simple arrangements will be denoted by $\NSMA{n}$, $\NSAMA{n}$, $\NSMCh{n}$, and $\GSM{k}{n}$.

\subsection{Mutations}
A \emph{mutation} is a local transformation of an arrangement $\Gamma$ that only destroys, or creates,  or inverts a fan of $\Gamma$;
more precisely, it  is a homotopy of arrangements in which only one curve $\gamma$ moves, reaching or leaving or first reaching and then leaving 
a single vertex of the remaining arrangement $\Gamma\smallsetminus\{\gamma\}$ (Fig.~\ref{mutation}). 

\begin{figure}[!htb]
\centering
\psfrag{w}{\small moving curve}
\psfrag{if}{\small I-mutation}
\psfrag{df}{\small D-mutation}
\psfrag{cf}{\small C-mutation}

\includegraphics[width=0.75\linewidth]{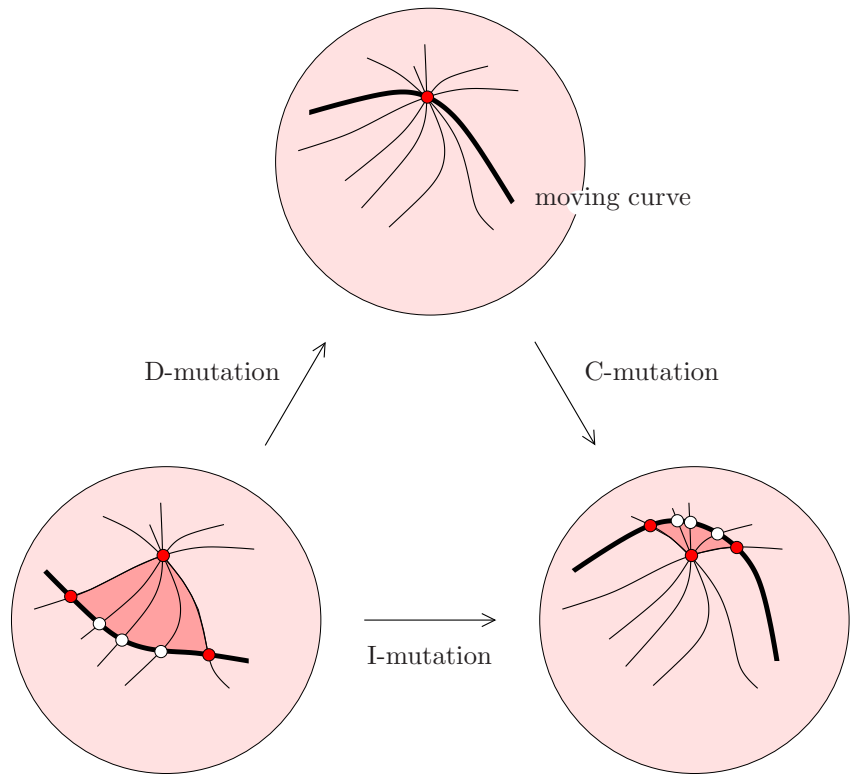}
\caption{A mutation destroys, creates, or inverts a fan supported by the moving curve.
(D-,C- and I- for destroying, creating and inverting.) \label{mutation}}
\end{figure}

As for pseudoline arrangements iterating this local transformation is in fact sufficient to obtain all possible arrangements of given order: 
{\it any two arrangements of the same order are homotopic via a finite sequence of mutations followed by an isotopy; in other words 
the graph of mutations on isomorphism classes of arrangements of given order is connected.}
The graph of mutations on simple arrangements of three double pseudolines is depicted in Figure~\ref{fig:mutation}. 

\begin{figure}[!htb]
\centering
\footnotesize
\tiny
\psfrag{8}{} \psfrag{7}{} \psfrag{6}{} \psfrag{5}{} \psfrag{4}{} \psfrag{3}{} \psfrag{2}{}
\psfrag{o24}{24} \psfrag{o12}{12} \psfrag{o2}{2} \psfrag{o4}{4} \psfrag{o6}{6} \psfrag{o1}{1}

\psfrag{A}{$04$}
\psfrag{B}{$07$}
\psfrag{C}{$18$}
\psfrag{CCone}{$18_1$}
\psfrag{D}{$25$} 
\psfrag{F}{$07$} 
\psfrag{G}{$37$}
\psfrag{H}{$15$}
\psfrag{HCone}{$15_1$}
\psfrag{J}{$43$}
\psfrag{JCone}{$43_1$}
\psfrag{K}{$25$}
\psfrag{KCone}{$25_1$}
\psfrag{L}{$33$}
\psfrag{LCone}{$33_1$}
\psfrag{M}{$32$}
\psfrag{N}{$25$}
\psfrag{Nstar}{$25^*$}
\psfrag{NstarCone}{$25^*_1$}
\psfrag{NstarCtwo}{$25^*_2$}
\psfrag{O}{$32$}
\psfrag{OCone}{$32_1$}
\psfrag{OCtwo}{$32_2$}
\psfrag{P}{$22$}
\psfrag{PCone}{$22_1$}
\psfrag{Q}{$25$}
\psfrag{R}{$36$}
\psfrag{Z}{$64$}
\includegraphics[width=.99\linewidth]{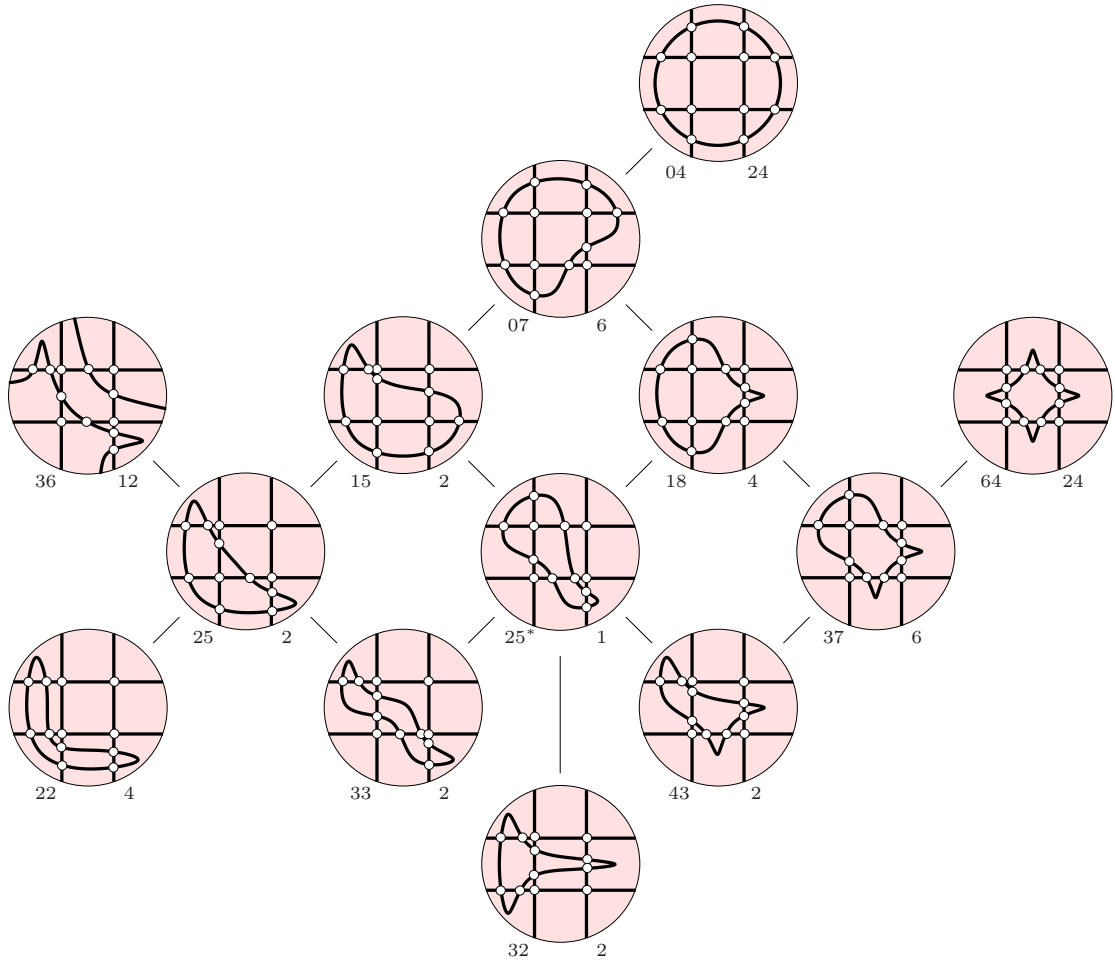}
\caption{The graph of mutations on simple arrangements of three double pseudolines.}
\label{fig:mutation}
\end{figure}

This connectedness result is proved in~\cite{G-hp-adp-06} 
by reduction to the case of arrangements of pseudolines.  The argument consists in mutating each double pseudoline of the arrangement until it becomes 
{\it thin} (that is, until there remains no vertex of the arrangement inside its enclosed M{\"o}bius strip), and to observe that arrangements of thin double pseudolines behave exactly as
simple pseudoline arrangements.
The fact that an arrangement can be mutated until all its double pseudolines become thin is ensured by the following crucial lemma: 
\begin{lemma}[Pumping Lemma~\cite{G-hp-adp-06}]\label{pl} 
Let $\Gamma$ be a simple arrangement of double pseudolines and let $\gamma$ be a distinguished double pseudoline of $\Gamma$.
Assume that there is a vertex
of $\Gamma$ lying in the interior of the
M{\"o}bius strip $\MS{\gamma}$ bounded  by $\gamma$.
Then there is a triangular face of $\Gamma$ supported by $\gamma$ and included in $\MS{\gamma}$. \qed
\end{lemma}
Observe that the Pumping Lemma allows to turn any arrangement of double pseudolines into a marked arrangement with any prescribed numbers 
of pseudolines missing the double pseudolines 
(Fig.~\ref{pumping}). 
 
\begin{figure}[!htb]
\includegraphics[width=0.90\linewidth]{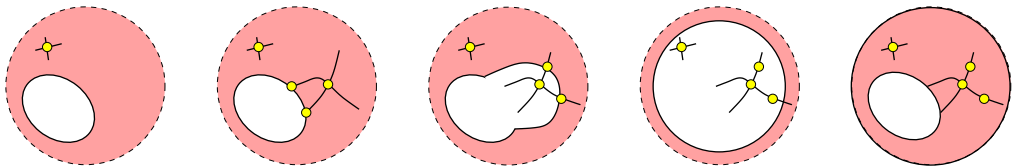}
\caption{Pumping vertices out of the M{\"o}bius strip of a double pseudoline and turning an arrangement into a marked one.\label{pumping}}
\end{figure}

From the connectedness property of mutation graphs it ensues an enumeration algorithm that simply consists in traversing the graph of mutations 
starting from any given arrangement. 
This  algorithm is sufficient in practice for the enumeration of (not necessarily simple) arrangements of three or four double pseudolines but already fails 
for simple arrangements of five double pseudolines because of random-access memory limitations.

In order to go a little bit further (and particularly, to enumerate arrangements of five double pseudolines), we  show that the mutation graph on the space of 
one-extensions of any given arrangement is connected.  
From this result we derive a simple enumeration algorithm for the class of arrangements of order $n+1$ based on 
the traversal of the graphs of mutations of the one-extension spaces of the arrangements of order $n$ whose working space is only 
proportional to the number of isomorphism classes of double pseudoline arrangements of order $n$ times the space of encoding of an arrangement.
It turns out that this incremental algorithm is sufficient in practice for the enumeration of simple arrangements of five double pseudolines.

 
\subsection{Geometric representation theorem} A (real two-dimensional) {\it projective geometry} is a topological projective point-line incidence geometry $(\pp,\lpp)$ 
whose point space~$\pp$ is a projective plane 
and whose line space $\lpp$ is a subspace of the space of pseudolines of~$\pp$; as usual the dual of a point $p$ of a projective geometry is denoted $p^*$ 
and is defined as its set of incident lines.

The {\it duality principle} for projective geometries asserts that the dual $(\lpp,\pp^*)$ of a projective geometry $(\pp,\lpp)$ is still a projective geometry,
i.e., $\lpp$ is a projective plane and~$\pp^*$ is a subspace of the space of pseudolines of $\lpp$; in particular the dual of a finite point set is an arrangement of pseudolines.
 The {\it  \GRT} for pseudoline arrangements asserts that the converse is true:  
{\it any arrangement of pseudolines is isomorphic to the dual arrangement of a finite  set of points of a 
projective geometry}. (This is an easy consequence of the duality principle for projective geometries combined with the embedabbility of any arrangement of pseudolines 
in the line space of a projective geometry~\cite{gpwz-atp-94}, see the discussion in~\cite{G-hp-adp-06}.)

The {\it \GRT}  for pseudoline arrangements has the following extension for marked arrangements of double pseudolines:  {\it 
first, the dual of a convex body of a projective geometry---defined as its set of tangent lines, i.e., the set of lines touching the body but not its interior---is 
a double pseudoline in the line space of the projective geometry; 
second, the dual arrangement of a finite set of pairwise disjoint marked convex bodies is a marked arrangement of double pseudolines; third, 
 any marked arrangement of double pseudolines is isomorphic to the dual arrangement 
of a finite set of pairwise disjoint marked convex bodies of a projective geometry~\cite{G-hp-adp-06}.} 
In other words: {\it the class of  isomorphism classes of marked arrangements of double pseudolines coincides with the class of isomorphism classes of 
dual arrangements of finite families of disjoint marked convex bodies of projective geometries.}
Similarly the {\it \GRT} for  marked M{\"o}bius arrangements asserts that their class coincides with the class of dual arrangements of 
finite families  of disjoint marked convex bodies of affine  geometries~\cite{G-hp-adp-06}.

From this result, it is easy, using continuous motions, to derive the following extension of the Pumping Lemma.
\begin{lemma} \label{evpl}
Let $\Gamma$ be an arrangement of double pseudolines and let $\gamma \in \Gamma$.
Assume that there is a vertex of $\Gamma$ lying in the interior of the
M{\"o}bius strip $\MS{\gamma}$ bounded  by $\gamma$ and that the vertices of $\Gamma$ supported by $\gamma$ are simple (i.e., have degree four).
Then there is at least two fans of $\Gamma$ included in $\MS{\gamma}$ and supported by $\gamma$. \qed
\end{lemma}
This enhanced version of the Pumping Lemma allows to mutate any double pseudoline until it becomes thin while keeping fixed any one of its points. 
Replacing the resulting thin double pseudoline by one of its core pseudoline and putting back in the arrangement the mutated double pseudoline   
we get a T-marked version of the initial arrangement with  one pseudoline touching the arrangement at the point kept fixed (Fig.~\ref{pumpingrevised}).  

\begin{figure}[!htb]
\includegraphics[width=0.90\linewidth]{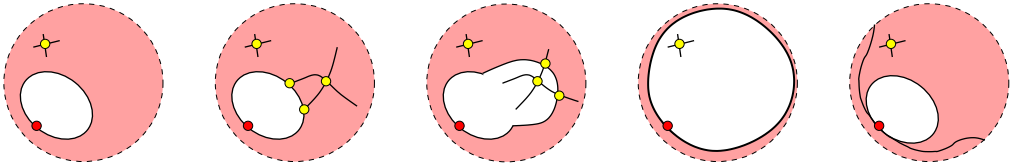}
\caption{Pumping vertices out of the M{\"o}bius strip of a double pseudoline while keeping fixed one of the points of the double pseudoline and turning an arrangement into 
a T-marked one.}
\label{pumpingrevised}
\end{figure}

By a repeated application of this T-marking operation one can turn any double pseudoline arrangement into a T-marked version
 with any prescribed pattern of touching points. In particular if we choose as set of touching points the 
set of pairs $(\gamma, x)$ where $\gamma$ ranges over the set of double pseudolines of the arrangement $\Gamma$ and where $x$ ranges over 
the set $V_\gamma$ of  vertices of the arrangement lying on $\gamma$ we obtain a 
{\it linearization} of the arrangement $\Gamma$, i.e., 
an arrangement of pseudolines $\ell(\gamma,x)$, $x \in V_\gamma$, $\gamma \in \Gamma$,  
such that 
(1) the sub-arrangement of  $\lgx$, $x \in V_\gamma$, is a cyclic arrangement with a distinguished central cell $\CA{\gamma}$ 
(the one including the topological disk $\DS{\gamma}$), 
(2)  the family of boundaries of the 
$\CA{\gamma}$, $\gamma \in \Gamma$, is an arrangement of double pseudolines 
isomorphic to $\Gamma$ (Fig.~\ref{linearization}).

\begin{figure}[!htb]
\psfrag{x}{$x$}
\psfrag{gamma}{$\gamma$}
\psfrag{lgx}{$\ell(\gamma,x)$}
\psfrag{gp}{$\gamma'$}
\psfrag{lgpx}{$\ell(\gamma',x)$}

\centering
\includegraphics[width=0.900000075000\linewidth]{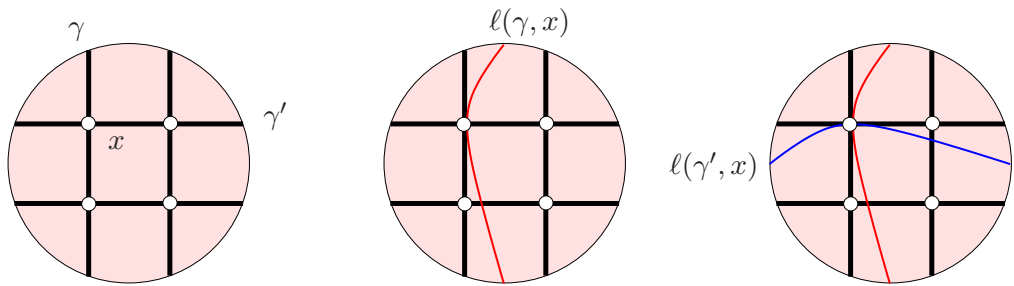}
\caption{Two steps of the linearization process of an arrangement of two double pseudolines. \label{linearization}}
\end{figure}

If we think of our double pseudoline arrangement $\Gamma$ as the dual of a configuration of smooth  convex bodies $\Delta$, 
this linearization process corresponds to insert a point $p(\gamma, x)$ on the intersection between $x$ and the convex body $\delta$ 
 corresponding to $\gamma$ and to replace $\delta$ by the convex hull $C_\delta$ of the $p(\gamma, x)$, $x \in\gamma$, with respect to any line 
avoiding $\delta$. The dual arrangement of the convex bodies $C_\delta$ is then isomorphic to the arrangement $\Gamma$.  
Concatening this linearization algorithm with any  algorithmic version of the 
\GRT\ for pseudoline arrangements (e.g.~\cite{as-plada-05}) we thus get an algorithmic version of the \GRT\
for double pseudoline arrangements (Fig.~\ref{cvgrt}).     

\begin{figure}[!htb]
\centering
\psfrag{duality}{$\stackrel{\mathrm{duality}}{\longrightarrow}$}
\psfrag{linearbis}{$\stackrel{\mathrm{linearization}}{\longrightarrow}$}
\psfrag{duality}{${\longrightarrow}$}
\psfrag{linearbis}{${\longrightarrow}$}
\psfrag{linear}{linearization}
\psfrag{grtpoint}{GRT for pseudoline arrang.}
\includegraphics[width=0.90000\linewidth]{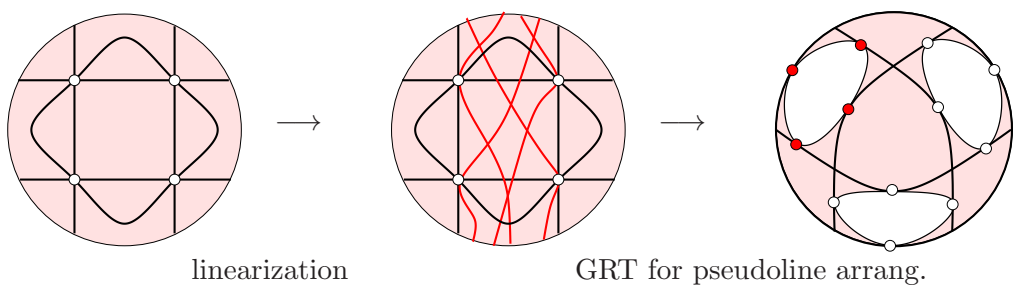}
\caption{Concatening the linearization process of an arrangement with a constructive version of the \GRT\ for pseudoline 
arrangements yields to a constructive version of the \GRT\ for double pseudoline arrangements. \label{cvgrt}}
\end{figure}


The enhanced version of the Pumping Lemma will be used both 
in our enumeration algorithm---to insert a double pseudoline in a given double pseudoline arrangement---and 
in the proof of the connectedness of the space of one-extensions of a given arrangement, proof to which we now come.


\section{Connectedness of the spaces of one-extensions}\label{sec:connectedness}

A one-extension of an arrangement of $n$ double pseudolines $\Gamma$ is an
arrangement of $n+1$ double pseudolines $\Gamma'$ of which $\Gamma$ is a
subarrangement.  The double pseudoline of $\Gamma'$ not in $\Gamma$ is called 
the one-extension element.  

For example Figure~\ref{nodesoneextensionthr}  depicts representatives of the twenty-three  simple (isomorphism classes of) one-extensions of 
an arrangement of two double pseudolines, with the curved double pseudoline in the role of the one-element extension.
\begin{figure}[!htb]
\footnotesize
\def\factor{0.15315015000023}
\centering
\psfrag{8}{\small 8} \psfrag{7}{\small 7} \psfrag{6}{\small 6} \psfrag{5}{\small 5} \psfrag{4}{\small 4} \psfrag{3}{\small 3} \psfrag{2}{\small 2}
\psfrag{8}{} \psfrag{7}{} \psfrag{6}{} \psfrag{5}{} \psfrag{4}{} \psfrag{3}{} \psfrag{2}{}
\psfrag{CC1}{$18_1$}
\psfrag{HC1}{$15_1$}
\psfrag{JC1}{$43_1$} \psfrag{KC1}{$25_1$} \psfrag{LC1}{$33_1$} \psfrag{NstarC1}{$25^*_1$} 
\psfrag{NstarC2}{$25^*_2$} \psfrag{OC1}{$32_1$} \psfrag{OC2}{$32_2$} \psfrag{PC1}{$22_1$}

\psfrag{A}{$04$}
\psfrag{B}{$07$}
\psfrag{C}{$18$}
\psfrag{CCone}{$18_1$}
\psfrag{D}{$25$} 
\psfrag{F}{$07$} 
\psfrag{G}{$37$}
\psfrag{H}{$15$}
\psfrag{HCone}{$15_1$}
\psfrag{J}{$43$}
\psfrag{JCone}{$43_1$}
\psfrag{K}{$25$}
\psfrag{KCone}{$25_1$}
\psfrag{L}{$33$}
\psfrag{LCone}{$33_1$}
\psfrag{M}{$32$}
\psfrag{N}{$25$}
\psfrag{Nstar}{$25^*$}
\psfrag{NstarCone}{$25^*_1$}
\psfrag{NstarCtwo}{$25^*_2$}
\psfrag{O}{$32$}
\psfrag{OCone}{$32_1$}
\psfrag{OCtwo}{$32_2$}
\psfrag{P}{$22$}
\psfrag{PCone}{$22_1$}
\psfrag{Q}{$25$}
\psfrag{R}{$36$}
\psfrag{Z}{$64$}
\psfrag{o24}{24} \psfrag{o12}{12} \psfrag{o2}{2} \psfrag{o4}{4} \psfrag{o6}{6} \psfrag{o1}{1}
\psfrag{o24}{} \psfrag{o12}{} \psfrag{o2}{} \psfrag{o4}{} \psfrag{o6}{} \psfrag{o1}{}
\includegraphics[width = \factor\linewidth]{P049000}
\includegraphics[width = \factor\linewidth]{P073300}
\includegraphics[width = \factor\linewidth]{P181030}
\includegraphics[width = \factor\linewidth]{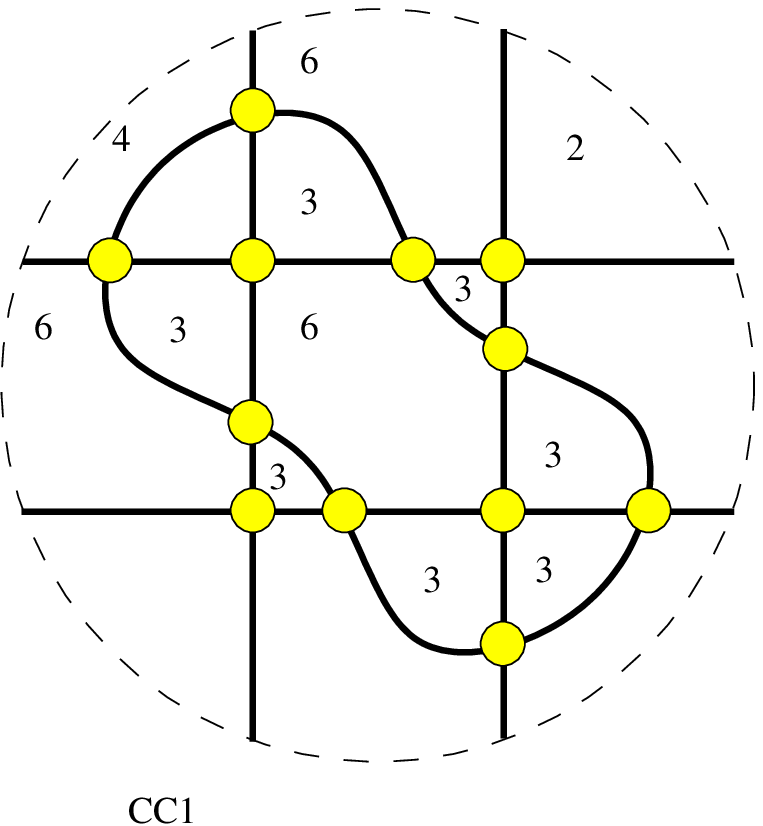}
\includegraphics[width = \factor\linewidth]{P370003}
\includegraphics[width = \factor\linewidth]{P154300}
\includegraphics[width = \factor\linewidth]{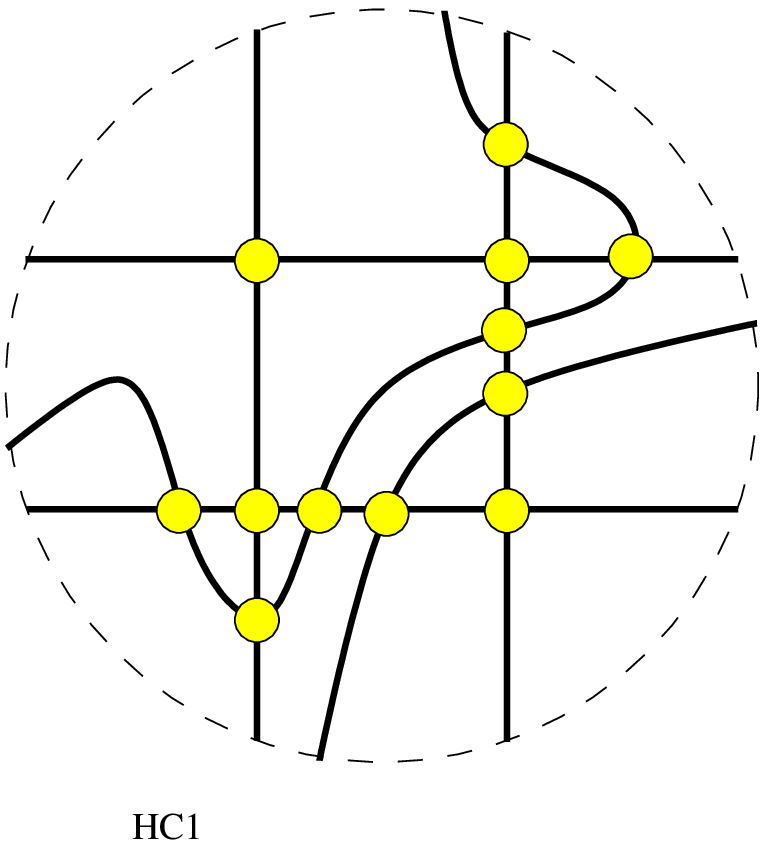}
\includegraphics[width = \factor\linewidth]{P433021}
\includegraphics[width = \factor\linewidth]{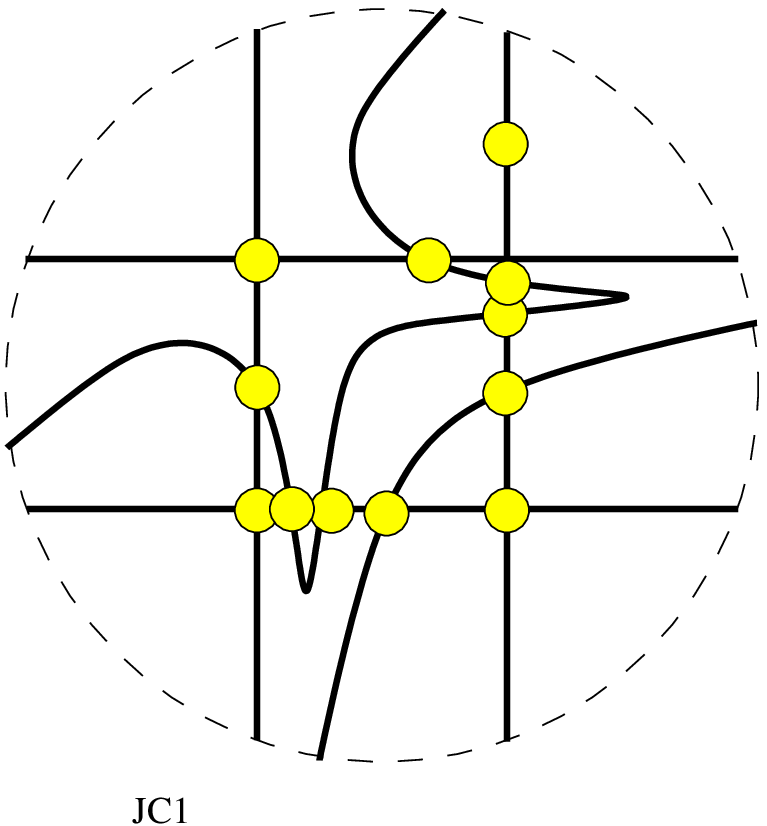}
\includegraphics[width = \factor\linewidth]{P228010}
\includegraphics[width = \factor\linewidth]{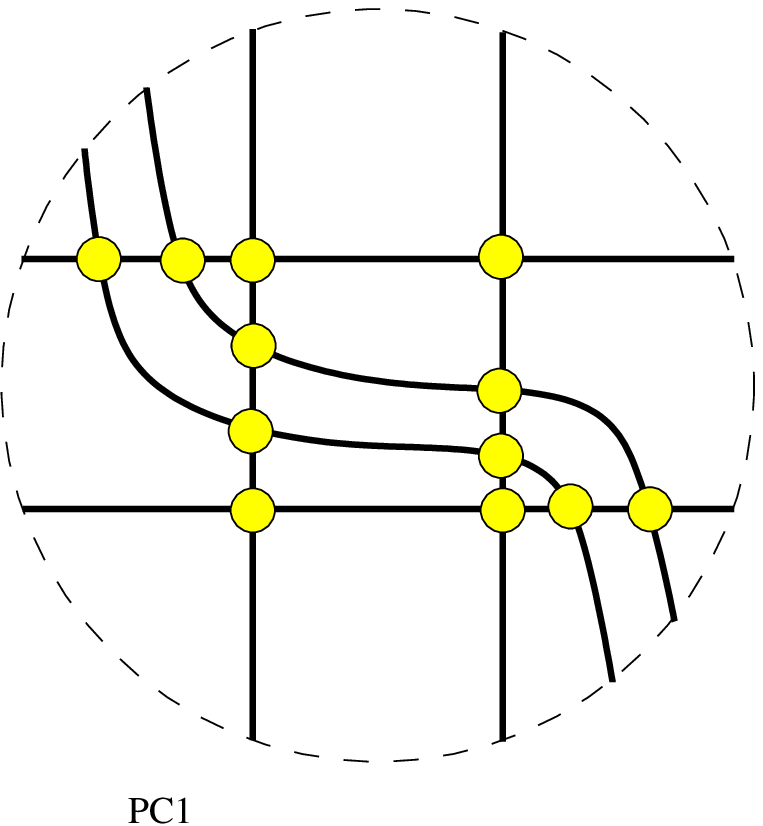}
\includegraphics[width = \factor\linewidth]{P333310}
\includegraphics[width = \factor\linewidth]{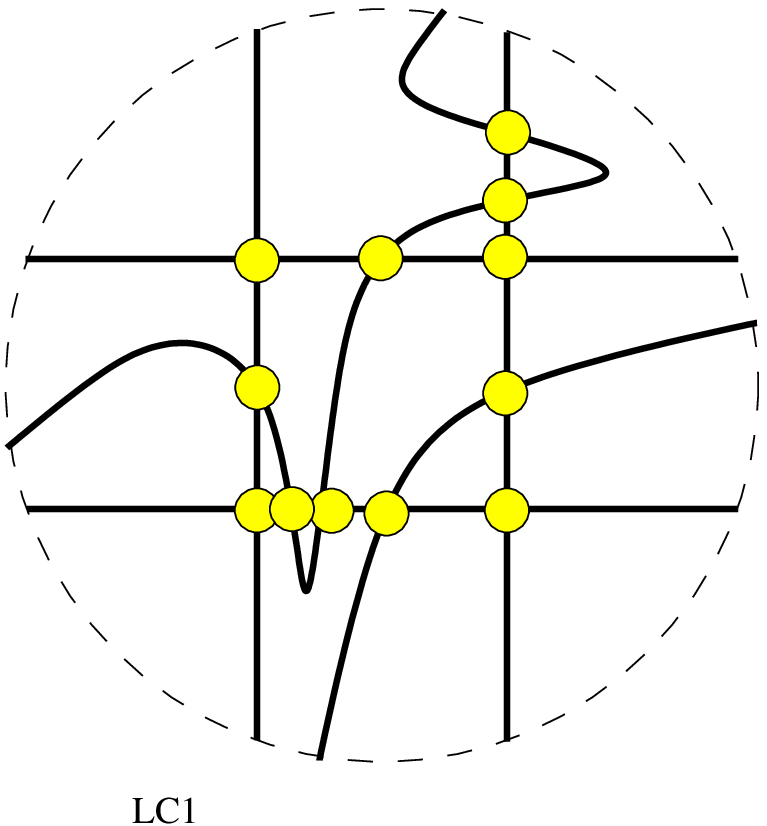}
\includegraphics[width = \factor\linewidth]{P326020}
\includegraphics[width = \factor\linewidth]{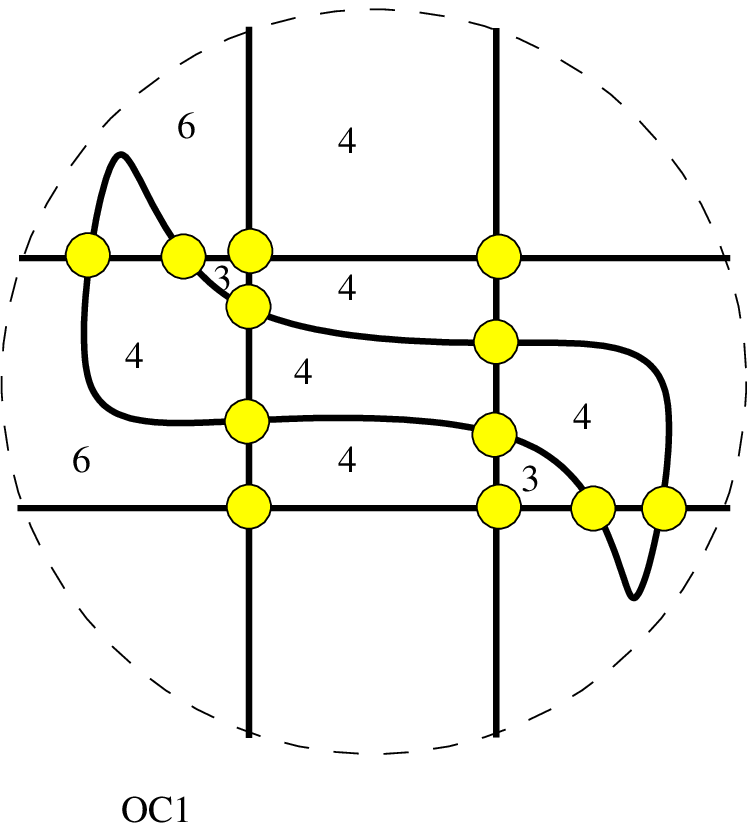}
\includegraphics[width = \factor\linewidth]{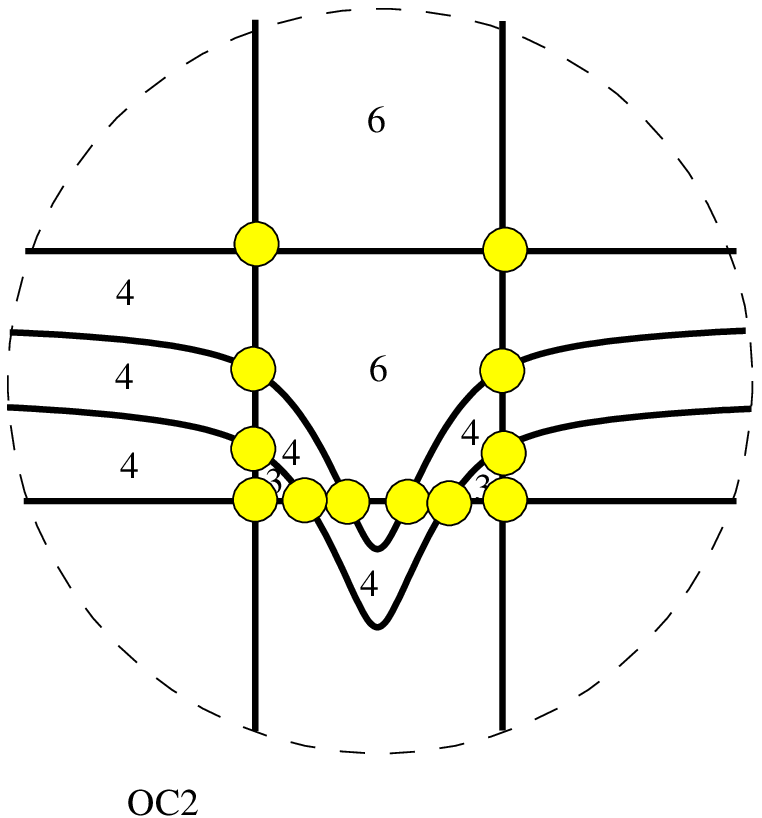}
\includegraphics[width = \factor\linewidth]{P252310}
\includegraphics[width = \factor\linewidth]{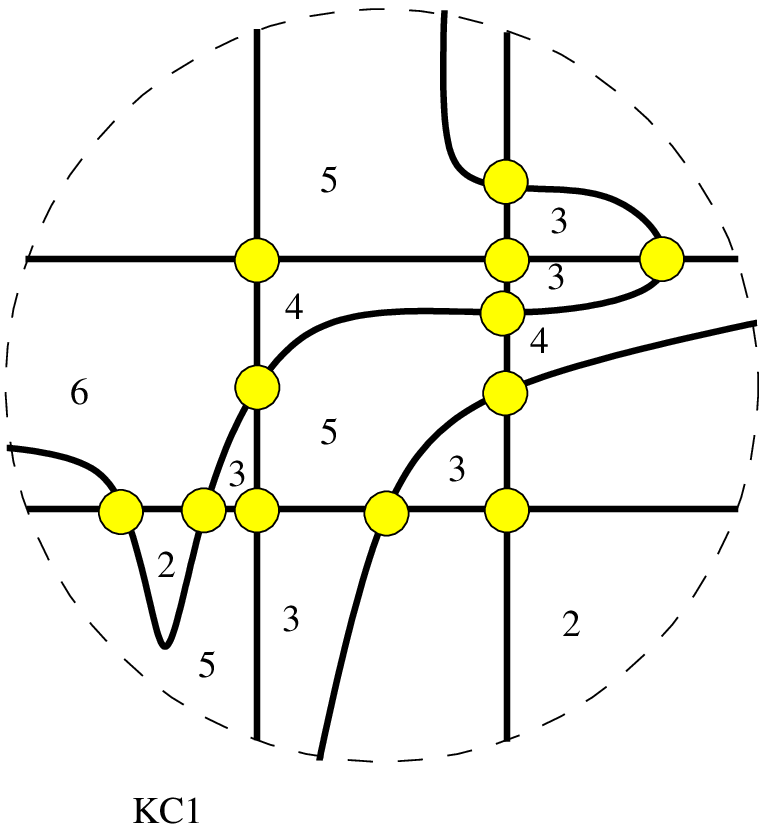}
\includegraphics[width = \factor\linewidth]{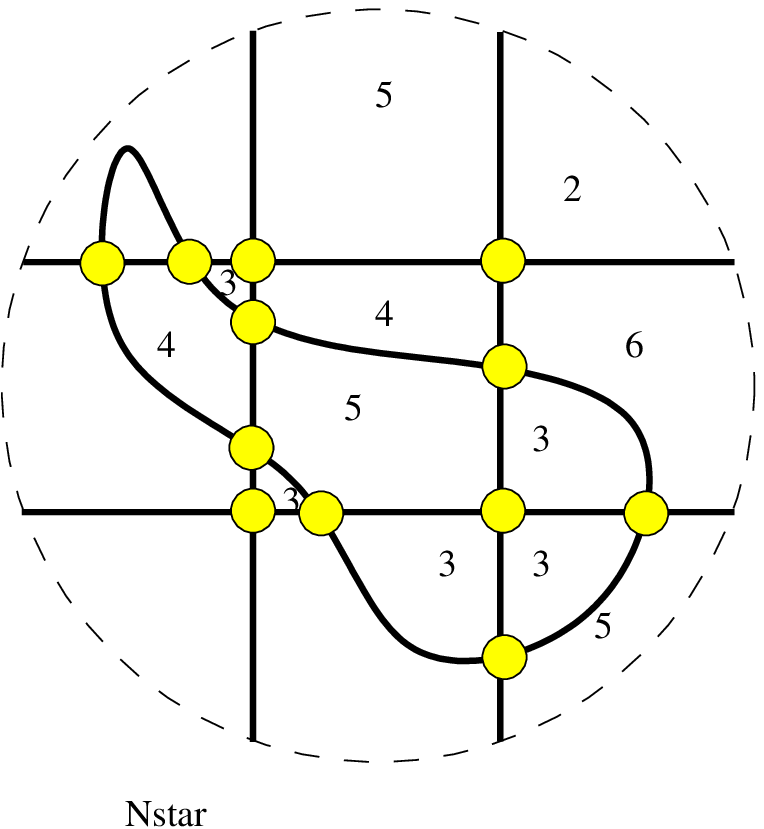}
\includegraphics[width = \factor\linewidth]{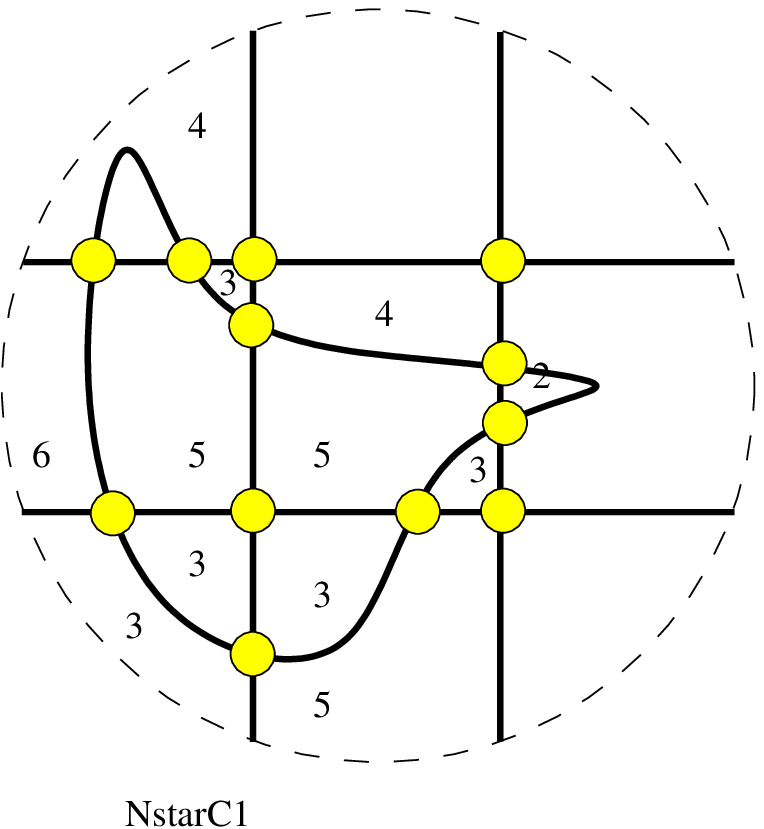}
\includegraphics[width = \factor\linewidth]{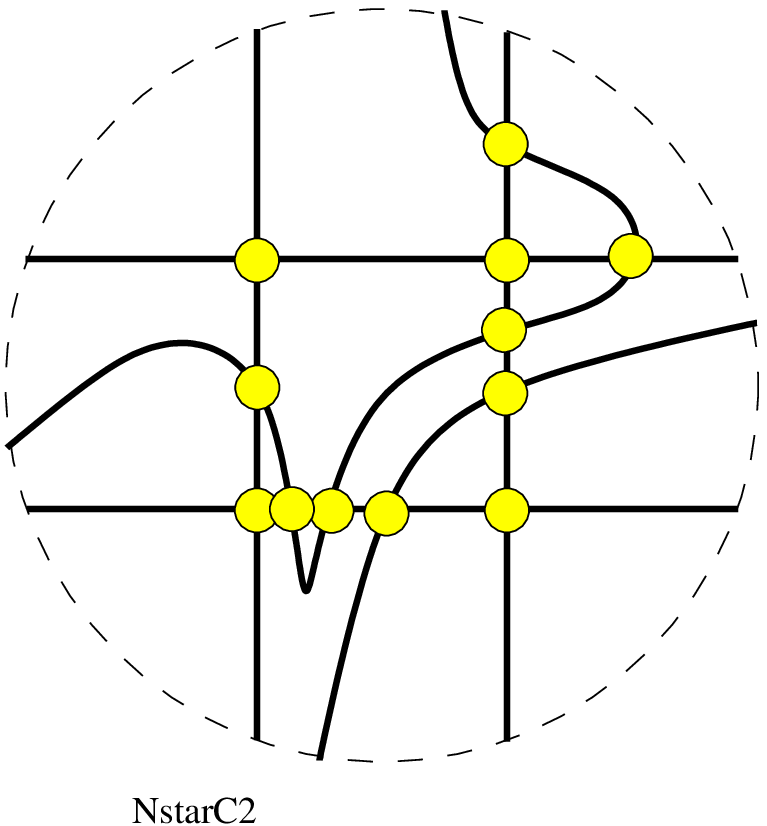}
\includegraphics[width = \factor\linewidth]{P360040}
\includegraphics[width = \factor\linewidth]{P6400003}
\caption{Representatives of the twenty-three  isomorphism classes of simple one-extensions of an arrangement of two double pseudolines with the curved double pseudoline in  the role of the one-extension element. 
\label{nodesoneextensionthr}}
\end{figure}

\begin{theorem} \label{main}
Let $\Gamma'$ and $\Gamma''$ be two one-extensions of an arrangement of double pseudolines $\Gamma$.  
Then $\Gamma'$ and $\Gamma''$ are homotopic via a finite sequence of mutations followed by an isotopy
during which the only moving curves are the one-extension elements of $\Gamma'$ and~$\Gamma''$. 
\end{theorem}
\begin{proof}

Let $\gamma$ be a double pseudoline of $\Gamma$ and  let $\gp$ and $\gpp$  be the one-extension elements of $\Gamma'$ and $\Gamma''$, respectively.
Using continuous motions---thanks to the {\it \GRT}---one can easily reduce the analysis to the following case
\begin{enumerate}
\item $\gp$ is a {\it thin} double pseudoline in $\Gamma'$, i.e., a  double pseudoline whose M{\"o}bius strip is free of vertices of $\Gamma$;
\item $\gamma$ and $\gp$ are {\it touching} at $\sigma'$: by this we mean that $\sigma'$ is one the two 2-cells of size two of the subarrangement $\{\gamma,\gp\}$ 
and that $\sigma'$ is also a 2-cell of the whole arrangement $\Gamma'$;
\item similarly $\gpp$ is a thin double pseudoline in $\Gamma''$, and $\gamma$ and $\gp$ are touching at $\sigma''$;
\item $\gp$ and $\gpp$ coincide in the disk $\DS{\gamma}$, and $\sigma' = \sigma''$;
\item $\MS{\gp}$ is a tubular neighbourhood of a pseudoline $\gpstar$  with the property that $\gpstar$ intersects any double pseudoline of $\Gamma$ 
in exactly two points, the intersection points being transversal;
similarly $\MS{\gpp}$ is a tubular neighbourhood of a pseudoline $\gppstar$  with the property that $\gppstar$ intersects any double pseudoline of $\Gamma$ 
in exactly two points, the intersection points being transversal;
\item $\gpstar$ and $\gppstar$  coincide in $\DS{\gamma}$, and intersect finitely many in $\MS{\gamma}$, the intersection points being transversal;
\item $\MS{\gp} \cup \MS{\gpp}$ is a tubular neighbourhood of $\gpstar \cup \gppstar$.
\end{enumerate}
Let $\Gamma'_* = \Gamma \cup \{\gpstar\}$ and, similarly, let $\Gamma''_* = \Gamma \cup \{\gppstar\}$. It should be clear to the reader that the proof of our theorem 
boils down to show that the mixed arrangements $\Gamma'_*$ and $\Gamma''_* $ are homotopic via a finite sequence of mutations during which the only moving curves 
are the pseudolines $\gpstar$ and $\gppstar$. This is precisely here that we are going to use Lemma~\ref{evpl}.  But before using this lemma  we define a \gcurve\ as a 
 connected component of the trace on $\MS{\gamma}$ of a double pseudoline of $\Gamma \setminus\{\gamma\}$, 
we observe that the pseudoline $\gpstar$ intersects a \gcurve\ in at most one point (necessarily a transversal intersection point), and similarly that 
the pseudoline $\gppstar$ intersects a \gcurve\ in at most one point.  Let now $\tau$ be a copy of $\gamma$ and let 
$\tauG$ be the symmetric difference between $\Gamma$ and $\{\gamma,\tau\}$. If $\tau$ is thin in $\tauG$ we are done 
modulo an isotopy. Otherwise Lemma~\ref{evpl} asserts that there exists a fan $\Delta$ of the arrangement $\tauG$ supported by $\tau$, included in $\MS{\tau}$, and 
missing the 2-cell $\sigma'$. Modulo a finite sequence of mutations in $\tauGp$ or in $\tauGpp$ or in both $\tauGp$ and in $\tauGpp$ during which the only moving curves are 
$\gpstar$ or $\gppstar$ or both $\gpstar$ and $\gppstar$ one can assume that $\Delta$ misses $\gpstar$ and $\gppstar$.  We then perform a mutation of $\Delta$ in $\tauG$ 
with $\tau$ in the role of the moving curve---during this process we do not touch at $\sigma'$. By repeated application of this process we arrive at the situation 
where $\tau$ is thin in $\tauG$ and where $\gpstar$  and $\gppstar$ are included in $\MS{\tau}$, excepted their common part in $\DS{\gamma}$---which is included in $\DS{\tau}$.  
At this point we are done modulo an isotopy. 
\end{proof}

\begin{figure}[!htb]
\centering
\footnotesize
\psfrag{A}{$04$}
\psfrag{B}{$07$}
\psfrag{C}{$18$}
\psfrag{CCone}{$18_1$}
\psfrag{D}{$25$} 
\psfrag{F}{$07$} 
\psfrag{G}{$37$}
\psfrag{H}{$15$}
\psfrag{HCone}{$15_1$}
\psfrag{J}{$43$}
\psfrag{JCone}{$43_1$}
\psfrag{K}{$25$}
\psfrag{KCone}{$25_1$}
\psfrag{L}{$33$}
\psfrag{LCone}{$33_1$}
\psfrag{M}{$32$}
\psfrag{N}{$25$}
\psfrag{Nstar}{$25^*$}
\psfrag{NstarCone}{$25^*_1$}
\psfrag{NstarCtwo}{$25^*_2$}
\psfrag{O}{$32$}
\psfrag{OCone}{$32_1$}
\psfrag{OCtwo}{$32_2$}
\psfrag{P}{$22$}
\psfrag{PCone}{$22_1$}
\psfrag{Q}{$25$}
\psfrag{R}{$36$}
\psfrag{Z}{$64$}
\includegraphics[width=.9999585750000000090000750\linewidth]{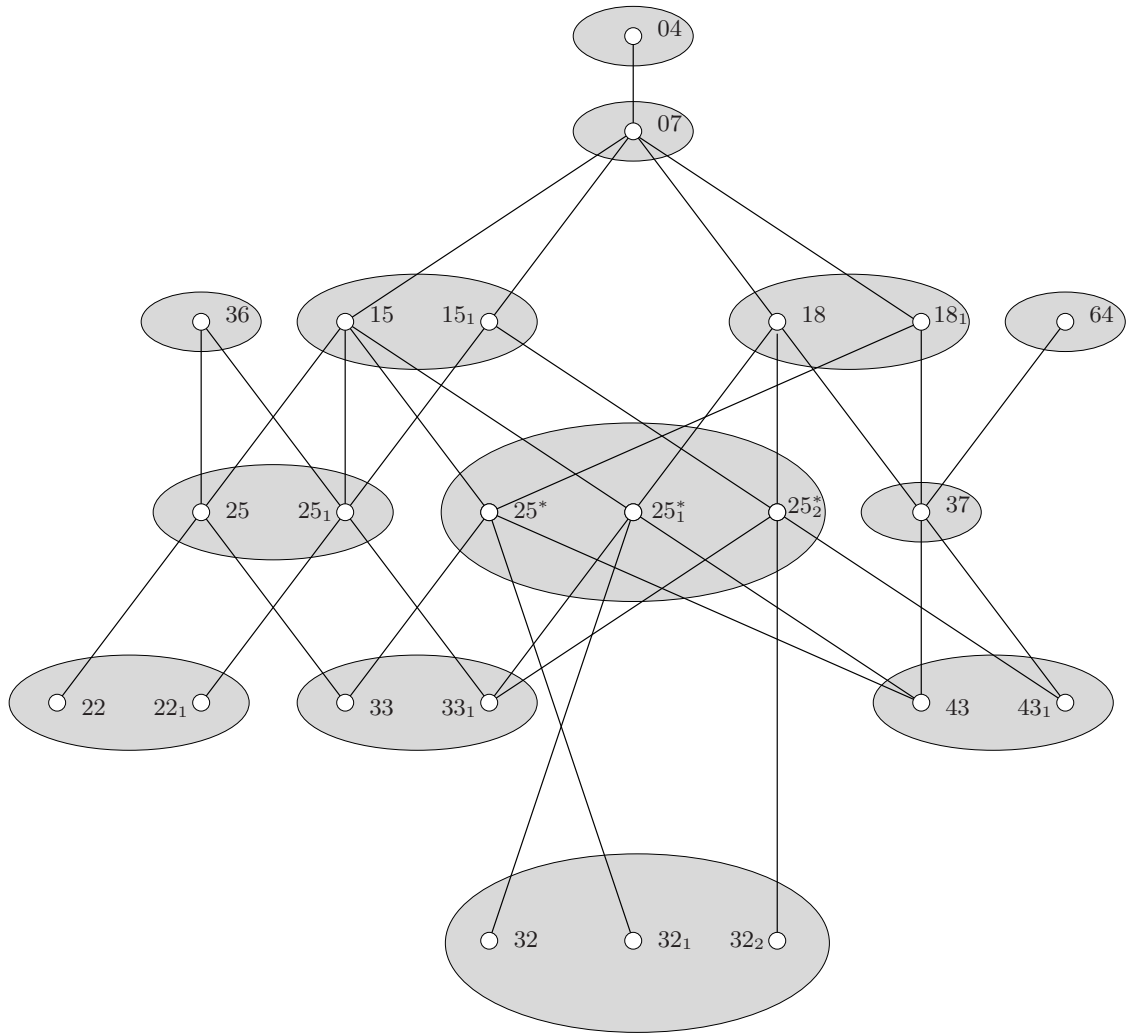}
\caption{The graph of mutations on the simple one-extensions of an arrangement of two  double pseudolines.}
 \label{fig:mutationter}
\end{figure}

In other words Theorem~\ref{main} asserts that the space of one-extensions of an arrangement of double pseudolines is connected under mutations.
Figure~\ref{fig:mutationter} depicts the graph of mutations  on simple one-extensions of an arrangement of two double pseudolines.
Similarly one can define a $k$-extension of an arrangement of $n$ double pseudolines $\Gamma$ as an
arrangement of $n+k$ double pseudolines $\Gamma'$ of which $\Gamma$ is a
subarrangement and show that the space of $k$-extensions of an arrangement of double pseudolines is connected under mutations. 
Finally we mention that similar results hold for marked arrangements and their $k$-extensions.

\section{The incremental algorithm}\label{sec:incremental}

\subsection{Description}
A \emph{pointed} arrangement is an arrangement with a distinguished double pseudoline. 
We always use the notation $A^{\bullet}$ for a pointed arrangement and $A$ for its non-pointed version (and similarly for sets of pointed arrangements). 
We also use the notation $\operatorname{Sub}(A)$ for the set of subarrangements of an arrangement $A$.

Let $\mathcal{A}_{n}$ denote the set of isomorphism classes of arrangements of $n$ double pseudolines, and $p_{n}$ be its cardinality.
Our algorithm enumerates $\mathcal{A}_{n}$ from an enumeration $a_1,a_2,\ldots, a_{p_{n-1}}$ of 
$\mathcal{A}_{n-1}$, by mutating an added distinguished double pseudoline to the~$a_i$.  

\begin{algorithm}
\caption{Incremental enumeration}
\begin{algorithmic}

\REQUIRE $\mathcal{A}_{n-1}=\{a_1,\ldots,a_{\NPA{n-1}}\}$.
\ENSURE  $\mathcal{A}_{n}$.

\medskip
\FOR{$i$ from $1$ to $\NPA{n-1}$}
	\STATE $A^{\bullet}\leftarrow\;$add a pointed double pseudoline to $a_i$.
	\IF{$\operatorname{Sub}(A) \cap \{a_1,\ldots, a_{i-1}\}=\emptyset$}
		\STATE write $A$.
	\ENDIF
	\STATE $Q^{\bullet}\leftarrow[A^{\bullet}]$. $S^{\bullet}\leftarrow\{A^{\bullet}\}$.
	\WHILE{$Q^{\bullet}\ne\emptyset$}
		\STATE $A^{\bullet}\leftarrow \mathrm{pop}\;Q^{\bullet}$.
		\STATE $T\leftarrow\;$list the fans of $A^{\bullet}$ supported by its pointed double pseudoline.
		\FOR{$t\in T$}
			\STATE $B^{\bullet}\leftarrow\;$mutate the fan $t$ in $A^{\bullet}$. 
			\IF{$B^{\bullet}\notin S^{\bullet}$} 
				\IF{$\operatorname{Sub}(B)\cap \{a_1,\ldots, a_{i-1}\}=\emptyset$ and ${B}\notin {S}$}
					\STATE write B.
				\ENDIF
				\STATE $ \mathrm{push}(B^{\bullet},Q^{\bullet})$. $S^{\bullet}\leftarrow S^{\bullet}\cup\{B^{\bullet}\}$.
			\ENDIF	
		\ENDFOR
	\ENDWHILE
\ENDFOR
\end{algorithmic}
\end{algorithm}


For each $i\in\{1,\ldots,\NPA{n-1}\}$, the algorithm enumerates the subset $S_i^{\bullet}$ of arrangements of $\mathcal{A}_{n}^{\bullet}$ containing $a_i$ as a subarrangement, 
by mutations of a distinguished added double pseudoline. From the set $S_i$, it selects the subset $R_i$ of arrangements with no subarrangements in $\{a_1,\ldots,a_{i-1}\}$.
In other words, $R_i$ is the subset of arrangements of $\mathcal{A}_{n}$ whose first subarrangement, in the sequence  $a_1,\ldots,a_{\NPA{n-1}}$, is $a_i$. 
Thus, $\mathcal{A}_{n}$ is the disjoint union of the~$R_i$.

An alternative approach{}\footnote{We thank Luc Habert for this suggestion.}
 for counting arrangements is to enumerate the subsets $S_i^{\bullet}$ and to compute, for each arrangement $A^{\bullet}$ of $S_i^{\bullet}$, the number $\sigma(A^{\bullet})$ of double pseudolines $\alpha$ of $A$ such that $A$ pointed at $\alpha$ is isomorphic to $A^{\bullet}$. Then
$$p_{n}=\frac{1}{n}\sum_{i=1}^{\NPA{n-1}} \sum_{A^{\bullet}\in S_i^{\bullet}} \sigma(A^{\bullet}).$$
The main advantage of this version is to reduce the number of accesses to the data base of arrangements of order $n-1$;
 an advantage which turns to be crucial in case the data base does not fit in main memory. 
However, it only counts $\NPA{n}$ and can not provide a data base for $\mathcal{A}_{n}$.


\subsection{Adding a double pseudoline}
One of the important steps of the incremental method is to add a double pseudoline to an initial arrangement. 
Our method uses three steps (see  Fig.~\ref{fig:addADPL}):
\begin{enumerate}
\item \textsl{duplicate a double pseudoline}: we choose one arbitrary double pseudoline $\gamma$, duplicate it, drawing a new double pseudoline $\gamma'$ 
completely included in the M{\"o}bius strip $\MS{\gamma}$ and we denote $R$ any rectangle delimited by $\gamma$ and $\gamma'$. 

\item \textsl{flatten}: we pump the double pseudoline $\gamma'$ 
such that no vertex of the arrangement lies in the M{\"o}bius strip $\MS{\gamma'}$. During this process, we do not touch the rectangle~$R$.

\item \textsl{add four crossings}: we replace the rectangle $R$ by four crossings between $\gamma$ and~$\gamma'$.  \end{enumerate}

\begin{figure}[!htb]
  \centering
  \includegraphics[width=0.75\linewidth]{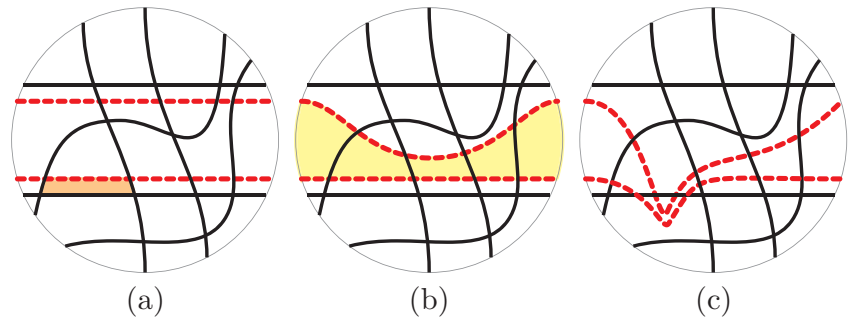}
  \caption{Three steps to insert a double pseudoline in a double pseudoline arrangement: duplicate a double pseudoline (a), flatten it (b) and add four crossings~(c).}
  \label{fig:addADPL}
\end{figure}

If we think of our double pseudoline arrangement as the dual of a configuration of convex bodies, this method corresponds to: (1) choosing one convex $C$ and drawing a new convex $C'$ inside $C$; (2) flattening the convex $C'$ until it becomes almost a single point, maintaining it almost in contact with the boundary of $C$; and (3) moving $C'$ outside $C$.

\subsection{Encoding an arrangement}
In order to manipulate arrangements, one can encode it in several different ways. 
We used two different encodings (one for an easy manipulation of it and one for short storage):
\begin{enumerate}
\item \textsl{flag representation}: 
a \emph{flag} of an arrangement $\Gamma$ is a triple $(v,e,f)$ consisting of a vertex $v$, an edge $e$ and a face $f$ of $\Gamma$, such that $v\in e\subset f$. 
The three involutions $\sigma_0,\sigma_1,\sigma_2$, that change the vertex, the edge and the face of a flag (cf. Fig.~\ref{fig:sigma}), completely determine $\Gamma$ 
(see, e.g.,~\cite{gt-tgt-01}).
This representation is convenient to perform all the necessary elementary operations we need to perform on arrangements (such as mutations, extensions, etc.).
\begin{figure}[!htb]
  \centering
  \includegraphics[width=.5\linewidth]{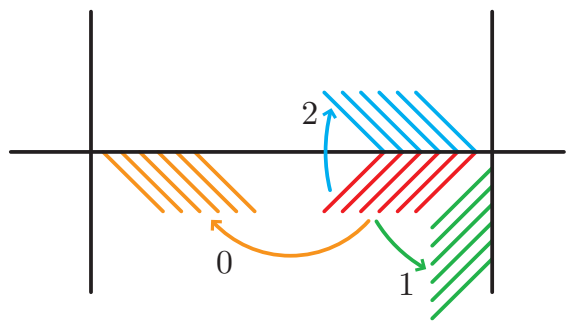}
  \caption{The three involutions $\sigma_0$, $\sigma_1$ and $\sigma_2$.}
  \label{fig:sigma}
\end{figure}

\item \textsl{encoding}: In the enumeration process, once we have finished to manipulate an arrangement $\Gamma$, we still have to remenber
it and to test whether the new arrangements we find afterwards are or not isomorphic to $\Gamma$. For this, we compute another representation of $\Gamma$, which is shorter than the flag representation, and which allows a quick isomorphism test.
 
We first associate to each flag $\phi=(v,e,f)$ of the arrangement $\Gamma$ a word $w_\phi$ constructed as follows. 
Let $\gamma_1$ be the double pseudoline containing $e$, and, for all $2\le p\le n$, let $\gamma_p$ be the $p$th  double pseudoline crossed by $\gamma_1$ on a walk starting at $\phi$ and oriented by $\sigma_0\sigma_1\sigma_2\sigma_1$.
This walk also defines a starting flag $\phi_i$ for each $\gamma_i$.
We walk successively on $\gamma_1,\ldots,\gamma_{n}$, starting from $\phi_i$ and in the direction given by $\sigma_0\sigma_1\sigma_2\sigma_1$, and index the vertices by $1,2,\ldots,V$ in the order of appearance. 
For all $i$, let $w_i$ denote the word formed by reading the indices of the vertices of $\gamma_i$ starting from $\phi_i$. 
The word $w_\phi$ is the concatenation of $w_1,w_2,\ldots,w_{n}$.
Finally, we associate to the arrangement $\Gamma$ the lexicographically smallest word among the $w_\phi$ where $\phi$ ranges over all the flags of $\Gamma$.
\end{enumerate}

\section{Results}\label{sec:results}
We have implemented this algorithm in the \texttt{C++}
 programming language. The documentation (as well as the source code) of this implementation is available 
upon request to the authors. 

\subsection{Numbers of arrangements and numbers of chirotopes}
This implementation provided us with the following table for the values of 
$\NSPA{n}$, $\NSPPA{n}$, $\NPA{n}$, $\NPPA{n}$,
$\NSOMPA{n}$, $\NOMPA{n}$,
$\NSPCh{n}$, $\NPCh{n}$, $\NSMA{n}$, $\NMA{n}$, 
$\NSMCh{n}$, and $\NMCh{n}$. 

\begin{table}[!ht]
$$
\begin{tabular}{c|cccrrr}
$n$  & $0$ & $1$ & $2$ & $3$ & $4$ & $5$  \\
\hline
$\NSPA{n}$ & $1$ & $1$ & $1$ & $13$ & $6\,570$ & $181\,403\,533$\\
$\NSPPA{n}$ & $0$ & $0$ & $0$ & $1$ & $615$ & \nco\\
$\NPA{n}$ & $1$ & $1$ & $1$ & $46$ & $153\,528$  & \nco\\
$\NPPA{n}$ & $0$ & $0$ & $0$ & $5$ & $18\, 648$  & \nco\\
$\NSOMPA{n}$ & $0$ & $1$ & $1$ & $67$ & $355\, 153$ & \nco\\
$\NOMPA{n}$ & $0$ & $1$ & $1$ & $398$ & \nco & \nco\\
$\NSPCh{n}$ & $1$ & $1$ & $1$ & $214$ & $2\,415\,112$ & \tba\\
$\NPCh{n}$ & $1$ & $1$ & $1$ & $1\,086$ & $58\,266\,120 $ & \tba\\
\hline
\hline
$\NSMA{n}$ & $1$ & $1$ & $1$ & $16$ & $11\,502$ & $238\, 834\, 187$\\
$\NSAMA{n}$ & $1$ & $1$ & $1$ & $22$ & $22\, 620$ & \nco\\
$\NMA{n}$ & $1$ & $1$ & $1$ & $59$ & $245\, 351$ & \nco\\
$\NAMA{n}$ & $1$ & $1$ & $1$ & $92$ & $488\, 303$ & \nco\\
$\NSMCh{n}$ & $1$ & $1$ & $1$ & $118$ & $541\, 820$ & \tba \\
$\NMCh{n}$ & $1$ & $1$ & $1$ & $531$ & $11\, 715\, 138$ & \nco\\
\end{tabular}
$$
\caption{Numbers of arrangements and numbers of chirotopes.}
\label{maincount}
\end{table}
Apart the trivial cases of $n=0,1,2$ only the values of 
$\NSPA{3}, \NSPPA{3}, \NSMA{3}, \NSAMA{3}, \NSMA{4}, \NSAMA{4}$ and $\NSMCh{3}$ were known previously.
The comparaison of the data of the 
lines $\NMA{n}$ and $\NAMA{n}$ confirms the intuition that very few M\"obius arrangements are isotopic to their mirror image. 

\subsection{Mixed arrangements}
In our presentation, we concentrated on the enumeration of double pseudoline arrangements or, equivalently  by the {\it \GRT},  on the enumeration of dual arrangements of configurations 
of pairwise disjoint convex bodies of projective (or affine) geometries. 
It is also interesting to enumerate dual arrangements of configurations of 
disjoint points and convex bodies. 
This  class is captured  by the notion of {\it mixed arrangements} and can be enumerated using the same algorithm.

\begin{figure}[!htb]
\centering
\includegraphics[width=0.75\linewidth]{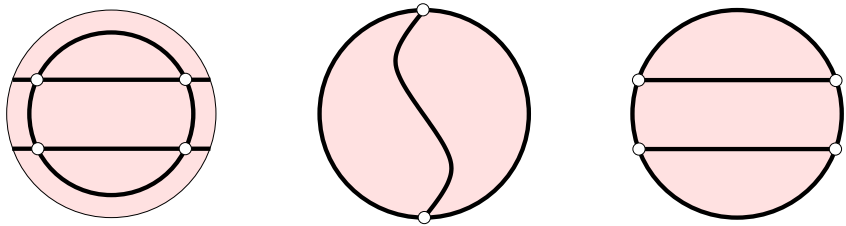}
\caption{\protect The three mixed arrangements of size two.  \label{addp}}
\end{figure}

A {\it mixed arrangement} in the projective plane $\pp$
is a finite set of pseudolines and double pseudolines in $\pp$  such that 
(i) any two pseudolines have a unique intersection point; 
(ii)~a pseudoline and a double pseudoline have exactly two intersection points and cross transversally at these points; and 
(iii) any two double pseudolines have exactly four intersection points, cross transversally at these points, and induce a cell decomposition of $\pp$ (Fig.~\ref{addp}). 
The {\it \GRT} for  mixed arrangements in a projective plane asserts that their class coincides with the class of dual arrangements of 
configurations of disjoint points and convex bodies of projective geometries~\cite{G-hp-adp-06}. 
A {\it mixed arrangement} in the M{\"o}bius strip $\MOB$ is a mixed arrangement in the projective plane $\OPC$  with the property that the intersection of
the topological disks surrounded by the pseudolines  and double pseudolines of the arrangement is nonempty and 
contains the point at infinity. 
The {\it \GRT} for  mixed arrangements in a M{\"o}bius strip asserts that their class coincides with the class of dual arrangements of 
configurations of disjoint points and convex bodies of affine geometries~\cite{G-hp-adp-06}.

Our incremental enumeration algorithm for double pseudoline arrangements 
extends easily to mixed arrangements in the projective plane and in the M{\"o}bius strip.
Our implementation provided us with the following four tables for the numbers of isomorphism classes of mixed arrangements composed of $n$ 
pseudolines and $m$ double pseudolines in the projective plane or in the M{\"o}bius strip. 
The sequences of $p^S_{n,0}$, $q^S_{n,0}$, $p_{n,0}$ and $q_{n,0}$ appear in the database of 
The On-Line Encyclopedia of Integer Sequences{}\footnote{\url{http://www.research.att.com/~njas/sequences/A006248}}    
under the code numbers A006248, A006247, A063800, and A063854 respectively.

\begin{table}[!ht]
$$
\begin{tabular}{c|cccccc}
$p^S_{n,m}$  & $0$ & $1$ & $2$ & $3$ & $4$ & $5$  \\
\hline
$0$ & $1$ & $1$ & $1$ & $13$ & $6\,570$ & $181\,403\,533$\\
$1$  &   $1$  & $1$  & $4$   & $626$    & $4\,822\,394$ \\
$2$  & $1$   & $2$  &  $48$  &  $86\,715$ \\
$3$  & $1$  & $5$ &  $1\,329$ \\
$4$  & $1$  & $25$ & 80\,253  \\
$5$  &  $1$ & 302  \\
6  & 4   &  9\,194 &    &   &       &\\
7  & 11   & 556\,298  &    &   &       &\\
8  & 135   &   &    &   &       &\\
9  & 4\,382   &   &    &   &      &\\
10  & 312\,356   &   &    &   &      &\\
\end{tabular}
$$
\caption{Simple projective mixed arrangements.}
\label{smpacount}
\end{table}

\begin{table}[!htb]
$$
\begin{tabular}{c|cccccc}
$q^S_{n,m}$  & $0$ & $1$ & $2$ & $3$ & $4$ & $5$  \\
\hline
$0$ & $1$ & $1$ & $1$ & $16$ & $11\,502$ & $238\, 834\, 187$\\
$1$  &   $1$  & $1$  & $7$   & $1\,499$    & $9\, 186\,477$ \\
$2$  & $1$   & $3$  &  $140$  &  $245\,222$ \\
$3$  & $1$  & $13$ &  $5\,589$ \\
$4$  & $2$  & $122$ & $416\,569$  \\
$5$  &  $3$ & $2\,445$  \\
6  & $16$   &  $102\,413$ &    &   &       &\\
7  & $135$   & $7\, 862\,130$  &    &   &       &\\
8  & $3\, 315$   &   &    &   &       &\\
9  & $158\,830$   &   &    &   &      &\\
\end{tabular}
$$
\caption{Simple M{\"o}bius mixed arrangements. 
}
\end{table}

\begin{table}[!htb]
$$
\begin{tabular}{c|ccccc}
$p_{n,m}/q_{n,m}$  & $0$ & $1$ & $2$ & $3$ & $4$ \\
\hline
$0$ & $1/1$ & $1/1$ & $1/1$ & $46/59$ & $153\,528/245\,351$ \\
$1$  & $1/1$  & $1/1$  & $9/17$   & $6\,998/15\,649$    &  \\
$2$  & $1/1$   & $3/5$  &  $265/799$  &   \\
$3$  & $1/1$  & $16/45$ &  $18\,532/74\,559$ \\
$4$  & $2/3$  & $159/832$ &  \\
$5$  &  $4/11$ & $4\,671/37\,461$  \\
6  & $17/93$   &  $342\,294$/\nco &    &   &       \\
7  & $143/2\, 121$   &  &    &   &       \\
8  & $4\, 890/122\, 508$   &   &    &   &       \\
9  & $461\,053$/\nco   &   &    &   &      \\
\end{tabular}
$$
\caption{Projective/M{\"o}bius  mixed arrangements.}
\end{table}

\clearpage
\subsection{Automorphism groups} Our implementation  provided us with the following values for the sequences $(k,\GSP{k}{n})$, 
$k \in  \{ u \mid \GSP{u}{n}\neq 0\}$, in the case $n=2,3$, and~$4$
$$(8,1)$$
$$(1,1) (2,5) (4,2) (6,2) (12,1) (24,2)$$
$$(1,6\,042) (2,466) (3,15) (4,26) (6,14) (8,1) (12,3) (16,1) (24,2)$$ 
from which we have deduced that $\NSPCh{2} = 1$, $\NSPCh{3} = 214$ and $\NSPCh{4}= 2\,415\,112.$
The corresponding sequences for not necessarily simple projective arrangements, 
simple M\"obius arrangements, and not necessarily simple M\"obius arrangements are in order

$$(8,1)$$
$$(1,10) (2,20) (4,5) (6,7) (12,1) (24,3)$$
$$(1,150\,042) (2,3\,288) (3,58) (4,84) (6,45) (8,2) (12,5) (16,1) (24,3),$$ 

$$(4,1)$$
$$(1,6) (2,6) (4,2) (6,2)$$
$$(1,11\,088) (2,387) (3,6) (4,12) (6,7) (8,2),$$
and
$$(4,1)$$
$$(1,33) (2,20) (4,3) (6,3)$$
$$(1, 242\,815), (2, 2\,466), (3,16), (4,38), (6,13), (8,3).$$
These data confirm the intuition 
that most of the arrangements have trivial automorphism groups, excepted in the basic case of (projective or M\"obius) 
arrangements of three double pseudolines.  
For example only $528$ of the $6570$ simple projective arrangements of four double pseudolines have a non trivial automorphism  group.
The two simple projective  arrangements of  four double pseudolines with automorphism groups of maximal order (i.e., 24)  
are depicted in the Figure~\ref{fig:extremal}.

\begin{figure}[!htb]
  \centering
  \includegraphics[width=.75\columnwidth]{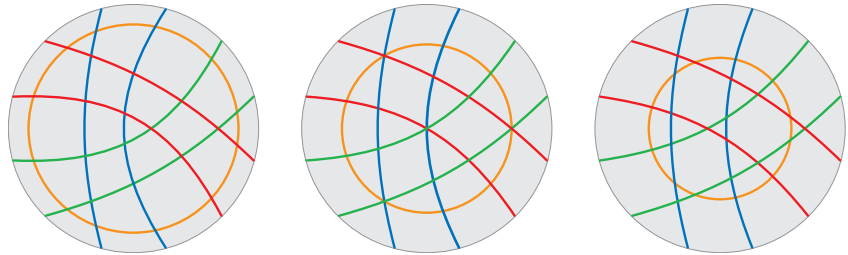}
  \caption{The three projective arrangements of four double pseudolines whose automorphism groups have the maximal size 24.}
  \label{fig:extremal}
\end{figure}

\subsection{Running time}
Let us briefly comment on running time. 
Observe that our algorithm can be parallelized very easily (separating each enumeration of $S_i$ and $R_i$, for $i\in\{1,\ldots,p_{n-1}\}$).
In order to obtain the number of simple projective arrangements of five double pseudolines we used 
four processors of $2$GHz for almost $3$ weeks: to say it differently each iteration of the {\bf for-to} loop
of our enumeration algorithm takes roughly twenty  minutes.  The working space was bounded by 
$\max |S_i^{\bullet}| = 279\, 882$ (times the space of the encoding of a single configuration, i.e., about 80 characters).  
Finding the number of not necessarily simple arrangements of five double pseudolines is doable using  a similar working space 
but much more process time.   

Finally, Fig.~\ref{fig:statistics} shows the evolution of the ratio between the sizes of the sets $R_i$ and $S_i^{\bullet}$ (during the enumeration of simple projective arrangements of five double pseudolines). We have also observed that $\sum |S^{\bullet}_i| / \sum |R_i|\simeq 5$, which confirms that hardly any configurations of five convex bodies have symmetries.

\begin{figure}[!htb]
  \centering
  \includegraphics[width=.75\columnwidth]{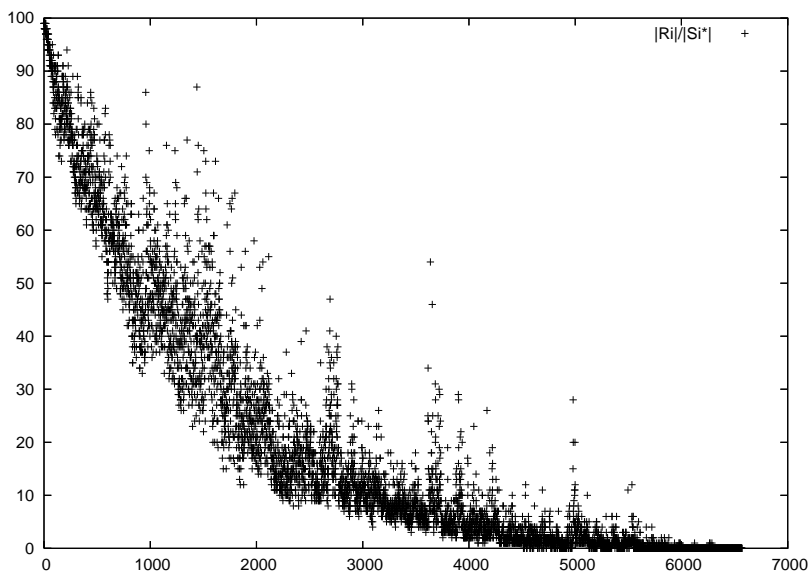}
  \caption{Percentage of new configurations.}
  \label{fig:statistics}
\end{figure}

\subsection{Statistics on mutation graphs.} 
We provide in Table~\ref{degreenewtwo} 
the distribution of the number $x_n$  of simple projective arrangements of four double pseudolines  with $n$ neighboords
in their mutation graph (the number of edges of this mutation graph is therefore 54544 and its average degree is 8). Still in Table~\ref{degreenewtwo}
we provide the distribution of the number $y_n$ of simple projective arrangements of four double pseudolines with $n$ triangular faces. 
In particular we see that there are two simple projective arrangements of four double pseudolines  with the maximum number (16) of triangular faces. 
They are 
depicted in Figure~\ref{fig:triangleextremal}.

\begin{table}[!htb]
$$
\begin{tabular}{c|cc}
$n$ & $x_n$ & $y_n$ 
\\\hline
1 &6 &0
\\\hline
2 &28& 0
\\\hline
3 &69 &9
\\\hline
4 &149& 69
\\\hline
5 &317& 219
\\\hline
6 &655& 566
\\\hline
7 &951& 942
\\\hline
8 &1288& 1336
\\\hline
9 &1228& 1306
\\\hline
10& 959& 1056
\\\hline
11& 587& 649
\\\hline
12& 240& 289
\\\hline
13& 79 &102
\\\hline
14& 13& 22
\\\hline
15& 1 &3
\\\hline
16& 0 &2
\end{tabular}
$$
\caption{Number $x_n$ of vertices of degree $n$ in the mutation graph of simple arrangements of four double pseudolines and
number $y_n$ of simple arrangements of four double pseudolines with $n$ triangular faces. \label{degreenewtwo}}
\end{table}

\begin{figure}[!htb]
  \centering
  \includegraphics[width=.90\columnwidth]{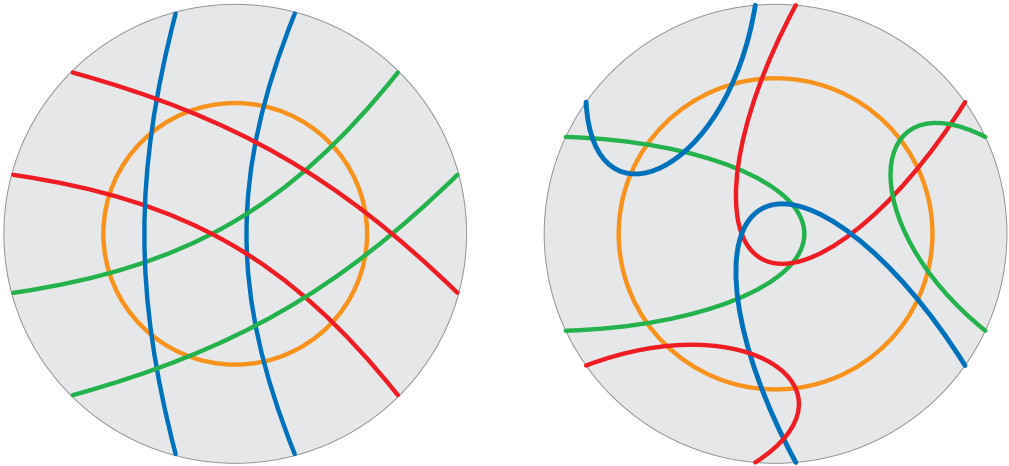}
  \caption{The two  simple projective arrangements of four double pseudolines  with the maximum number (16) of triangular faces.}
  \label{fig:triangleextremal}
\end{figure}

Finally, as a matter of example, we provide in Appendix~\ref{ghmthr}
the graph of D- and C-mutations on the space of arrangements of three double pseudolines. 
\clearpage
\section{Further developments}\label{sec:further}

The following questions and developments (among others) may be treated in a subsequent paper:
\begin{enumerate}
\item \textsl{Developing further implementation}: 
it will be interesting to add to our implementation the two following functionalities related to the {\it \GRT} : 
computing the dual arrangement of a configuration of convex bodies and conversely, computing a configuration of convex bodies 
whose dual arrangement is given. 

\item \textsl{Drawing an arrangement}: we have seen a method to add a pseudoline in an arrangement. Combined with a planar-graph-drawing algorithm, 
this provides an algorithm to draw an arrangement in the unit disk. 
For example,  the number of one-marked arrangements composed of $m$ double pseudolines can be interpreted as the number of 
drawings of the arrangements composed of $m$ double pseudolines with the property 
that the number of crossings between the arrangement and the boundary of the unit disk is minimum, i.e., $2m-2$.
Similarly, the number of mixed arrangements composed of one pseudoline and 
$m$ double pseudolines can be interpreted as the number of drawings of the arrangements composed of $m$ double pseudolines with the property 
that the number of crossings between the arrangement and the boundary of the unit disk is $2m$.

\item \textsl{Generation}: 
the {\it \AT} for double pseudoline arrangements affirms that the huge list of more than one hundred and eighty one millions of 
arrangements of at most five double pseudolines 
is an axiomatization of the class of double pseudoline arrangements. 
Due to its size this list seems hardly computationally workable for the generation of arrangements. 
However it is possible, as explained in~\cite{G-hp-adp-06},
 to reduce this axiomatization to a (relatively) short list of axioms  with a clear combinatorial or geometrical 
interpretation of each axiom in a maner very similar to the known simple axiomatizations of pseudoline arrangements~\cite{blswz-om-99,k-ah-92}.
This open the door to the generation of double pseudoline arrangements with prescribed properties 
using, for example, satisfiability solvers as proposed in~\cite{s-gomus-06,s-nrmvt-08}, 
see also~\cite{bg-gom-00,ff-gomgt-02,aak-eotsp-02,ff-cgspc-03,f-gtarg-01}. In particular it will be interesting to generate 
the arrangements that maximize the number of connected components of the intersection of the M\"obius strips surrounded by the double pseudolines.
If we think of the arrangement as the dual of a configuration of (disjoint) convex bodies this means that we seek for the configurations 
with the maximal number of connected components of transversals.
\item \textsl{Realizability}: it is well-known that certain pseudoline arrangements are not realizable in the standard projective geometry ${\mathcal P}^2(\mathbb{R})$. Inflating pseudolines into thin double pseudolines in such an arrangement give rise to non-realizable double pseudoline arrangements. Are there smaller examples? Are all arrangements of at most five double pseudolines realizable? 

\end{enumerate}
\section*{Acknowledgments}

We thank Luc Habert, \'Eric Colin de Verdi\`ere and Francisco Santos for interesting discussions on the subject (and technical support).





\appendix
\section{A graph of  D- and C-mutations \label{ghmthr}}

{\footnotesize
\begin{table}[p]
\begin{tabular}{c|cccccccccccccccc}
Node code & order autom. & fvector & nbr simple vertices & Neighbors \\ \hline 
\hline
\#12 & 24 &  60 &  0 &  \#6 \#20
\\ \hline
\hline
\#20 & 6 &  61 &  3 &  \#11 \#12 \#30
\\ \hline
\#15 &  6 &  34 &  3 &  \#8 \#22
\\ \hline
\#16 &  2 &  34 &  3 &  \#8 \#9 \#23 \#26
\\ \hline
\#18 &  2 &  43 &  3 &  \#9 \#11 \#25 \#28
\\ \hline
\#6 &  6 &  34 &  3 &  \#3 \#11 \#12
\\ \hline
\hline
\#27 &  2 &  34 &  6 &  \#17 \#35
\\ \hline
\#30 &  4 &  62 &  6 &  \#19 \#20 \#38
\\ \hline
\#8 &  2 &  24 &  6 &  \#4 \#13 \#15 \#16
\\ \hline
\#26 &  2 &  34 &  6 &  \#16 \#17 \#34
\\ \hline
\#28 &  1 &  44 &  6 &  \#17 \#18 \#19 \#36 \#37
\\ \hline
\#22 &  2 &  34 &  6 &  \#13 \#15 \#31
\\ \hline
\#23 &  1 &  33 &  6 &  \#13 \#14 \#16 \#33 \#34
\\ \hline
\#25 &  2 &  41 &  6 &  \#14 \#18 \#36
\\ \hline
\#9 &  1 &  25 &  6 &  \#4 \#5 \#14 \#16 \#17 \#18
\\ \hline
\#11 &  2 &  35 &  6 &  \#5 \#6 \#18 \#19 \#20
\\ \hline
\#3 &  4 &  16 &  6 &  \#1 \#5 \#6
\\ \hline
\#40 &  4 &  34 &  6 &  \#31 \#32
\\ \hline
\hline
\#38 &  6 &  63 &  9 &  \#29 \#30 \#45
\\ \hline
\#19 &  2 &  36 &  9 &  \#10 \#11 \#28 \#29 \#30
\\ \hline
\#36 &  1 &  42 &  9 &  \#24 \#25 \#28 \#44
\\ \hline
\#37 &  2 &  45 &  9 &  \#28 \#29 \#44
\\ \hline
\#17 &  1 &  26 &  9 &  \#9 \#10 \#24 \#26 \#27 \#28
\\ \hline
\#33 &  2 &  33 &  9 &  \#21 \#23 \#42
\\ \hline
\#31 &  2 &  35 &  9 &  \#21 \#22 \#39 \#40
\\ \hline
\#35 &  1 &  33 &  9 &  \#24 \#27 \#43
\\ \hline
\#34 &  1 &  33 &  9 &  \#23 \#24 \#26 \#42
\\ \hline
\#13 &  1 &  24 &  9 &  \#7 \#8 \#21 \#22 \#23
\\ \hline
\#14 &  1 &  24 &  9 &  \#7 \#9 \#23 \#24 \#25
\\ \hline
\#4 &  2 &  15 &  9 &  \#2 \#7 \#8 \#9
\\ \hline
\#5 &  2 &  17 &  9 &  \#2 \#3 \#9 \#10 \#11
\\ \hline
\#1 &  6 &  06 &  9 &  \#0 \#2 \#3
\\ \hline
\#32 &  2 &  24 &  9 &  \#21 \#40 \#41
\\ \hline
\hline

\#45 &  24 &  64 &  12 &  \#38
\\ \hline
\#29 &  6 &  37 &  12 &  \#19 \#37 \#38
\\ \hline
\#39 &  12 &  36 &  12 &  \#31
\\ \hline
\#0 &  24 &  04 &  12 &  \#1
\\ \hline
\#10 &  4 &  18 &  12 &  \#5 \#17 \#19
\\ \hline
\#2 &  6 &  07 &  12 &  \#1 \#4 \#5
\\ \hline
\#44 &  2 &  43 &  12 &  \#36 \#37
\\ \hline
\#24 &  1 &  25 &  12 &  \#14 \#17 \#34 \#35 \#36
\\ \hline
\#7 &  2 &  15 &  12 &  \#4 \#13 \#14
\\ \hline
\#21 &  2 &  25 &  12 &  \#13 \#31 \#32 \#33
\\ \hline
\#42 &  2 &  33 &  12 &  \#33 \#34
\\ \hline
\#43 &  2 &  32 &  12 &  \#35
\\ \hline
\#41 &  4 &  22 &  12 &  \#32
\end{tabular}
\caption{The graph of D- and C-mutations on the space of projective arrangements of three double pseudolines.  \label{demimutation}}
\end{table}
}

\begin{figure}[!htb]
\centering
\footnotesize
\tiny

\def\labA{$04$}
\def\labB{$07$}
\def\labC{$18$}
\def\labG{$37$}
\def\labH{$15$}
\def\labJ{$43$}
\def\labK{$25$}
\def\labL{$33$}
\def\labNstar{$25^*$}
\def\labO{$32$}
\def\labP{$22$}
\def\labR{$36$}
\def\labZ{$64$}

\def\labA{$\#0$}
\def\labB{$\#2$}
\def\labC{$\#10$}
\def\labG{$\#29$}
\def\labH{$\#7$}
\def\labJ{$\#44$}
\def\labK{$\#21$}
\def\labL{$\#42$}
\def\labNstar{$\#24$}
\def\labO{$\#43$}
\def\labP{$\#41$}
\def\labR{$\#39$}
\def\labZ{$\#45$}

\def\labZE{0}
\def\labON{1}
\def\labTW{2}
\def\labTH{3}
\def\labFO{4}
\def\labFI{5}
\def\labSI{6}
\def\labSE{7}
\def\labHE{8}
\def\labNI{9}

\def\labDZE{d0}
\def\labDON{d1}
\def\labDTW{d2}
\def\labDTH{d3}
\def\labDFO{d4}
\def\labDFI{d5}
\def\labDSI{d6}
\def\labDSE{d7}
\def\labDHE{d8}
\def\labDNI{d9}

\def\labVZE{v0}
\def\labVON{v1}
\def\labVTW{v2}
\def\labVTH{v3}
\def\labVFO{v4}
\def\labVFI{v5}
\def\labVSI{v6}
\def\labVSE{v7}
\def\labVHE{v8}
\def\labVNI{v9}

\def\labTZE{t0}
\def\labTON{t1}
\def\labTTW{t2}
\def\labTTR{t3}

\def\labON{$\#12$}
\def\labTW{$\#20$}
\def\labTH{$\#15$}
\def\labFO{$\#16$}
\def\labFI{$\#18$}
\def\labSI{$\#6$}
\def\labSE{$\#27$}
\def\labHE{$\#30$}
\def\labNI{$\#8$}

\def\labDZE{$\#26$}
\def\labDON{$\#28$}
\def\labDTW{$\#22$}
\def\labDTH{$\#23$}
\def\labDFO{$\#25$}
\def\labDFI{$\#9$}
\def\labDSI{$\#11$}
\def\labDSE{$\#3$}
\def\labDHE{$\#40$}

\def\labDNI{$\#38$}
\def\labVZE{$\#19$}
\def\labVON{$\#36$}
\def\labVTW{$\#37$}
\def\labVTH{$\#17$}
\def\labVFO{$\#33$}
\def\labVFI{$\#31$}
\def\labVSI{$\#35$}
\def\labVSE{$\#34$}
\def\labVHE{$\#13$}
\def\labVNI{$\#14$}
\def\labTZE{$\#4$}
\def\labTON{$\#5$}
\def\labTTW{$\#1$}
\def\labTTR{$\#32$}

\psfrag{0}{\labZE}
\psfrag{1}{\labON}
\psfrag{2}{\labTW}
\psfrag{3}{\labTH}
\psfrag{4}{\labFO}
\psfrag{5}{\labFI}
\psfrag{6}{\labSI}
\psfrag{7}{\labSE}
\psfrag{8}{\labHE}
\psfrag{9}{\labNI}
\psfrag{d0}{\labDZE}
\psfrag{d1}{\labDON}
\psfrag{d2}{\labDTW}
\psfrag{d3}{\labDTH}
\psfrag{d4}{\labDFO}
\psfrag{d5}{\labDFI}
\psfrag{d6}{\labDSI}
\psfrag{d7}{\labDSE}
\psfrag{d8}{\labDHE}
\psfrag{d9}{\labDNI}
\psfrag{v0}{\labVZE}
\psfrag{v1}{\labVON}
\psfrag{v2}{\labVTW}
\psfrag{v3}{\labVTH}
\psfrag{v4}{\labVFO}
\psfrag{v5}{\labVFI}
\psfrag{v6}{\labVSI}
\psfrag{v7}{\labVSE}
\psfrag{v8}{\labVHE}
\psfrag{v9}{\labVNI}
\psfrag{t0}{\labTZE}
\psfrag{t1}{\labTON}
\psfrag{t2}{\labTTW}
\psfrag{t3}{\labTTR}

\psfrag{t4}{\labA}    
\psfrag{t5}{\labB}
\psfrag{t6}{\labC}
\psfrag{t7}{\labG} 
\psfrag{t8}{\labH}
\psfrag{t9}{\labJ}
\psfrag{q0}{\labK}
\psfrag{q1}{\labL}
\psfrag{q2}{\labNstar}
\psfrag{q3}{\labO}
\psfrag{q4}{\labP}
\psfrag{q5}{\labR}
\psfrag{q6}{\labZ}
\includegraphics[width=.99\linewidth]{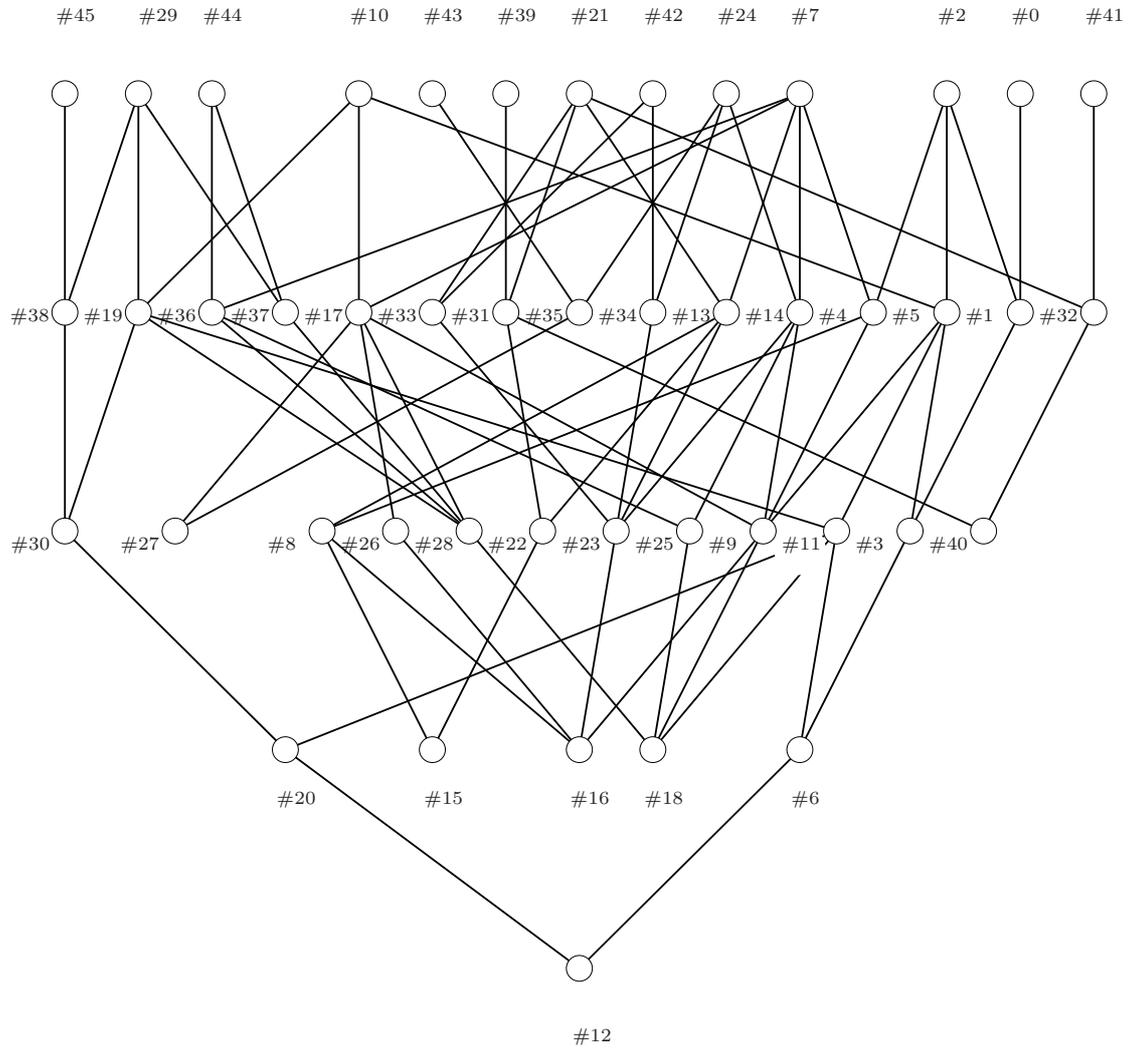}
\caption{The graph of D- and C-mutations on the space of projective arrangements of three double pseudolines.}
\label{fig:halfmutation}
\end{figure}

\end{document}